\documentclass[12pt]{article}
\usepackage{amsmath}
\usepackage{amsfonts}
\usepackage{amsthm}
\usepackage{latexsym}
\usepackage{amscd}
\usepackage{amssymb}
\usepackage{graphicx}
\usepackage{float}
\usepackage{scalerel}
\usepackage[applemac]{inputenc}
\usepackage{dsfont}
\usepackage{multirow}
\usepackage{dcolumn}
\usepackage{rccol}
\usepackage{enumitem}
 \usepackage[text={6.6in,8.8in},centering]{geometry} 
\usepackage{scalefnt}
\usepackage{booktabs}
\usepackage{subfigure}
\usepackage{caption}
\usepackage{color}
\usepackage[normalem]{ulem}
  \usepackage[overload]{empheq}
\usepackage[colorlinks=true]{hyperref}
\hypersetup{colorlinks=true, linkcolor=red, filecolor=blue, pagecolor=blue, urlcolor=green, citecolor=blue}

\usepackage{textcomp}
\usepackage{array}
\usepackage{supertabular}
\usepackage{hhline}

\usepackage{cancel}

\makeatletter
\newcommand\arraybslash{\let\\\@arraycr}
\makeatother
\newcolumntype{+}{>{\global\let\currentrowstyle\relax}}
\newcolumntype{^}{>{\currentrowstyle}}

\usepackage{setspace}
\newcommand\dsum{\displaystyle\sum}

\newlength{\bracewidth}

\def\fudge{\mathchoice{}{}{\mkern.5mu}{\mkern.8mu}}
\def\bbc#1#2{{\rm \mkern#2mu\vbar\mkern-#2mu#1}}
\def\bbb#1{{\rm I\mkern-3.5mu #1}}
\def\bba#1#2{{\rm #1\mkern-#2mu\fudge #1}}
\def\bb#1{{\count4=`#1 \advance\count4by-64 \ifcase\count4\or\bba A{11.5}\or
   \bbb B\or\bbc C{5}\or\bbb D\or\bbb E\or\bbb F \or\bbc G{5}\or\bbb H\or
   \bbb I\or\bbc J{3}\or\bbb K\or\bbb L \or\bbb M\or\bbb N\or\bbc O{5} \or
   \bbb P\or\bbc Q{5}\or\bbb R\or\bbc S{4.2}\or\bba T{10.5}\or\bbc U{5}\or
   \bba V{12}\or\bba W{16.5}\or\bba X{11}\or\bba Y{11.7}\or\bba Z{7.5}\fi}}
\def \R {\bb R}

\def \n {\noindent}

\def \diag {\textrm{diag}}
\restylefloat{float}

\newtheorem{theorem}{Theorem}[section]

\newtheorem{remark}{Remark}[section]
\newtheorem{lemma}{Lemma}[section]
\newtheorem{corollary}{Corollary}[section]

\title{Multi-Patch and Multi-Group Epidemic Models: A New Framework}

\author{Derdei Bichara\textsuperscript{1} and Abderrahman Iggidr\textsuperscript{2}\\
\textsuperscript{1} Department of Mathematics \& Center for Computational and
 Applied \\Mathematics, California State University, Fullerton, CA 92831, USA\\
\textsuperscript{2} Inria,  Universit\'e  de  Lorraine,  CNRS.  Institut  Elie  Cartan  de  Lorraine,\\  UMR  7502.
ISGMP Bat.  A, Ile du Saulcy, 57045 Metz Cedex 01, France.
}

\date{}

\begin{document}
\maketitle

\begin{abstract}
We develop a multi-patch and multi-group model that captures the dynamics of an infectious disease when the host is structured into an arbitrary number of groups and interacts into an arbitrary number of patches where the infection takes place.  In this framework, we model host mobility that depends on its epidemiological status, by a Lagrangian approach. This framework is applied to a general SEIRS model and the basic reproduction number $\mathcal R_0$ is derived. The effects of heterogeneity in groups, patches and mobility patterns on $\mathcal R_0$ and disease prevalence are explored. Our results show that for a fixed number of groups, the basic reproduction number increases with respect to the number of patches and the host mobility patterns. Moreover, when the mobility matrix of susceptible individuals is of rank one, the basic reproduction number is explicitly determined and was found to be independent of the latter if the matrix is also stochastic. The cases where mobility matrices are of rank one capture important modeling scenarios. Additionally, we study the global analysis of equilibria for some special cases. Numerical simulations are carried out to showcase the ramifications of mobility pattern matrices on disease prevalence and basic reproduction number.
\end{abstract}



\section{Introduction}
The role of heterogeneity in populations and their mobility have long been recognized as driving forces in the spread of infectious diseases \cite{AndMay91,dushoff1995effects,prothero1977disease,SattenspielSimon88}. Indeed, populations are composed of individuals with different immunological features and hence differ in how they can transmit or acquire an infection at a given time. These differences could result from demographic, host genetic or socio-economic factors \cite{AndMay91}. Populations also move across different geographical landscapes, importing their disease history with them either by infecting or getting infected in the host/visiting location.\\

While the concept of modeling epidemiological heterogeneity within a population goes back to Kermack and McKendrick in modeling the age of infection \cite{KmK1927}, the approach gained prominence with Yorke and Lajmonivich's seminal paper \cite{LajYo76} on the spread of gonorrhea, a sexually transmitted disease. An abundant and varied literature have followed on understanding the effects of  ``superspreaders " which are core groups on the disease dynamics \cite{BlytheCCC89,CCCBusen91,Jacq88,JacSimKoo95,Yorke:1978wu} or related multi-group models \cite{jmb, fall2007epidemiological,MR87c:92046,HuangCookeCCC92,Nold80,rushton1955deterministic,SattenspielSimon88} (and the references therein). Similarly, spatial heterogeneity in epidemiology has been extensively explored in different settings. Continuum models of dispersal have been investigated through diffusion equations \cite{metz2014dynamics} whereas \textit{islands} models have been dealt through metapopulation approach \cite{Arino08,ArinoPortet2015,ArVdd06,iggidr2014dynamics,IggidrSalletTsanou,SalVdd06,7606146}, defined here as continuous models with discrete dispersal.\\

Although the importance and the complete or partial analysis of these two types of heterogeneities have been studied separately in the aforementioned papers, little attention has been given to the simultaneous consideration of groups and spacial heterogeneities. Moreover, previous studies on multi-group rely on differential susceptibility in each group through the WAIFW (Who Acquires Infection From Whom \cite{AndMay91}) matrices which, we argue, are difficult to quantify. Similarly, in metapopulation (Eulerian) settings, the movement of individuals between patches is captured in terms of flux of population, making it nearly impossible to track the life-history of individuals after the interpatch mixing. \\

In this paper, we introduce a general modeling framework that structures populations into an arbitrary number of groups (e.g. demographic, ethnic or socio-economic grouping). These populations, with different health statuses,  spend certain amounts of time in an arbitrary number of locations, or \textit{patches}, where they could get infected or infect others. Each patch is defined by a particular risk of infection tied to environmental conditions of each patch. This approach allows us to track individuals of each group over time and to avoid the use of differential susceptibility of individuals or groups, which is theoretically nice but practically difficult to assess.  The likelihood of infection depends both on the time one spends (in a particular patch) and the risk associated with that patch. Moreover, we incorporate individuals' behavioral decisions through differential residence times. Indeed, individuals of the same group spend different amounts of time in different areas depending on their epidemiological conditions. We also considered two cases of the general framework, that are particularly important from modeling standpoint: when the susceptible and/or infected individuals of different groups have proportional residence times in different patches. That is, when the mobility matrix of susceptible (or infected) individuals, $\mathbb M$ (or $\mathbb P$) is of rank one. In these cases, we obtain explicit expressions of the basic reproduction number in terms of mobility patterns. It turns out that if $\mathbb M$ is of rank one and stochastic, the basic reproduction number is independent of the mobility patterns of susceptible host.\\

 In short, we address how group heterogeneity, or \textit{groupness}, patch heterogeneity, or \textit{patchiness}, mobility patterns and behavior each alter or mitigate disease dynamics. In this sense, our paper is a  direct extension of \cite{BicharaCCC2015,bichara2016dynamics,bichara2015sis,castillo2016perspectives} but also other studies that capture dispersal through Lagrangian approaches -- in which it is possible to track host movement after the interpatch mixing -- \cite{cosner2009effects,iggidr2014dynamics,rodriguez2001models,Ruktanonchai2016} and a recent paper \cite{falcon2016day} that investigates the effects of daily movements in the context of Dengue.\\

The paper is organized as follows. Section \ref{sec:modelderiv} explains the model derivation, states the basic properties and the computation of the basic reproduction number $\mathcal R_0(u,v)$ for $u$ groups and $v$ patches. Section \ref{sec:Heterogeneity} investigates the role of patch and group heterogeneity on the basic reproduction number, and how dispersal patterns alter  $\mathcal R_0(u,v)$ and the disease prevalence. Section \ref{Sec:GAS} is devoted to the existence, uniqueness and stability of equilibria for the considered system under certain conditions. Finally, Section \ref{sec:ConclusionDiscussions} is dedicated to concluding remarks and discussions.

\section{Derivation of the model}
\label{sec:modelderiv}

We consider a population that is structured in an arbitrarliy many $u$ groups interacting in $v$ patches. We consider a typical disease captured by an SEIRS structure. Naturally, $S_i$, $E_i$, $I_i$ and $R_i$ are the susceptible, latent, infectious and recovered individuals of Group $i$ respectively. The population of each group is denoted by $N_i=S_i+E_i+I_i+R_i$, for $i=1,\dots,u$. Individuals of Group $i$ spend on average some time in Patch $j$, $j=1,\dots,v$. The susceptible, latent, infected and recovered populations of group $i$ spend $m_{ij}$, $n_{ij}$, $p_{ij}$ and $q_{ij}$ proportion of times respectively in Patch $j$, for $j=1,\dots,v$. At time $t$, the \textit{effective} population of Patch $j$ is $N_j^{\textrm{eff}}=\sum_{k=1}^u(m_{kj}S_k+n_{kj}E_k+p_{kj}I_k+q_{kj}R_k)$. This \textit{effective} population of Patch $j$ describes the temporal dynamics of the population in Patch $j$ weighted by the mobility patterns of each group and each epidemiological status. Of this patch population,  $\sum_{k=1}^up_{kj}I_k$ are infectious. The proportion of infectious individuals in Patch $j$ is therefore,

$$\frac{\sum_{k=1}^up_{kj}I_k}{\sum_{k=1}^u(m_{kj}S_k+n_{kj}E_k+p_{kj}I_k+q_{kj}R_k)}$$

Susceptible individuals of Group $i$ could be infected in any Patch $j$, $j=1,\dots,v$ while visiting there. Hence, the dynamics of susceptible of Group $i$ is given by:

$$\dot S_i= \Lambda_i-\sum_{j=1}^v\beta_j m_{ij}S_i\frac{\sum_{k=1}^up_{kj}I_k}{\sum_{k=1}^u(m_{kj}S_k+n_{kj}E_k+p_{kj}I_k+q_{kj}R_k)}-\mu_iS_i+\eta_iR_i$$  

where $\Lambda_i$ denotes a constant recruitment of susceptible individuals of Group $i$, $\mu_i$ the natural death rate, $\beta_j$ the risk of infection and $\eta_i$ the immunity loss rate.  The patch specific risk vector $\mathcal B=(\beta_j)_{1\leq j\leq v}$ is treated as constant. However, in Subsection \ref{DDTransmission}, we also considered the case when this risk depends on the \textit{effective population size}.

The latent individuals of Group $i$ are generated through infection of susceptible and decreased by natural death and by becoming infectious at the rate $\nu_i$. Hence the dynamics of latent of Group $i$, for $i=1,\dots,u$, is given by:

$$\dot E_i=\sum_{j=1}^v\beta_j m_{ij}S_i\frac{\sum_{k=1}^up_{kj}I_k}{\sum_{k=1}^u(m_{kj}S_k+n_{kj}E_k+p_{kj}I_k+q_{kj}R_k)}-(\nu_i+\mu_i)E_i$$

The dynamics of infectious individuals of Group $i$ is given by
 
 $$\dot I_i=\nu_iE_i-(\gamma_i+\mu_i)I_i$$

where $\gamma_i$ is the recovery rate of infectious individuals. Finally, the dynamics of recovered individuals of Group $i$ is:

$$\dot R_i=\gamma_iI_i-(\eta_i+\mu_i)R_i$$

The complete dynamics of $u$-groups and $v$-patches SEIRS epidemic model is given by the following system:
\begin{equation} \label{PatchGen}
\left\{\begin{array}{llll}
\dot S_{i}= \Lambda_i-\sum_{j=1}^v\beta_jm_{ij}S_i\dfrac{\sum_{k=1}^up_{kj}I_k}{\sum_{k=1}^u(m_{kj}S_k+n_{kj}E_k+p_{kj}I_k+q_{kj}R_k)}-\mu_iS_i+\eta_iR_i,\\[3mm]
\dot E_i=\sum_{j=1}^v\beta_jm_{ij}S_i\dfrac{\sum_{k=1}^up_{kj}I_k}{\sum_{k=1}^u(m_{kj}S_k+n_{kj}E_k+p_{kj}I_k+q_{kj}R_k)}-(\nu_i+\mu_i)E_i\\
\dot I_i=\nu_iE_i-(\gamma_i+\mu_i+\delta_i)I_i\\
\dot R_i=\gamma_iI_i-(\eta_i+\mu_i)R_i
\end{array}\right.
\end{equation}
The description of parameters in Model (\ref{PatchGen}) is given in Table \ref{tab:Param}. These parameters are composed of three set of parameters: ecological/environmental (number of patches $v$ and their risk $\mathcal B$), epidemiological (Recruitment, death rates, recovery rate, etc) and behavioral (mobility matrices) parameters. A schematic description of the flow is given in Fig \ref{fig:FlowEnv}.

\begin{table}[h!]
  \begin{center}
    \caption{Description of the parameters used in System (\ref{PatchGen}).}
    \label{tab:Param}
    \begin{tabular}{cc}
      \toprule
      Parameters & Description \\
      \midrule
$\Lambda_i$ & Recruitment of the susceptible individuals in Group $i$\\
$\beta_j$ & Instantaneous risk of infection in Patch $j$\\
$\mu_i$ & \textit{Per capita} natural death rate of Group $i$\\
$\nu_{i}$ &  \textit{Per capita} rate at which latent in Group $i$ become infectious \\
$\gamma_i$  & \textit{Per capita} recovery rate of Group $i$\\
$m_{ij}$ &  Proportion of time susceptible individuals of Group $i$ spend in Patch $j$\\
$n_{ij}$ &  Proportion of time latent individuals of Group $i$ spend in Patch $j$\\
$p_{ij}$ &  Proportion of time infectious individuals of Group $i$ spend in Patch $j$\\
$q_{ij}$ &  Proportion of time recovered individuals of Group $i$ spend in Patch $j$\\
$\eta_{i}$  & \textit{Per capita} loss of immunity rate \\
$\delta_i$  & \textit{Per capita} disease induced death rate of Group $i$.\\
      \bottomrule
    \end{tabular}
  \end{center}
\end{table}
\begin{figure}[ht]
\centering
\includegraphics[scale =1]{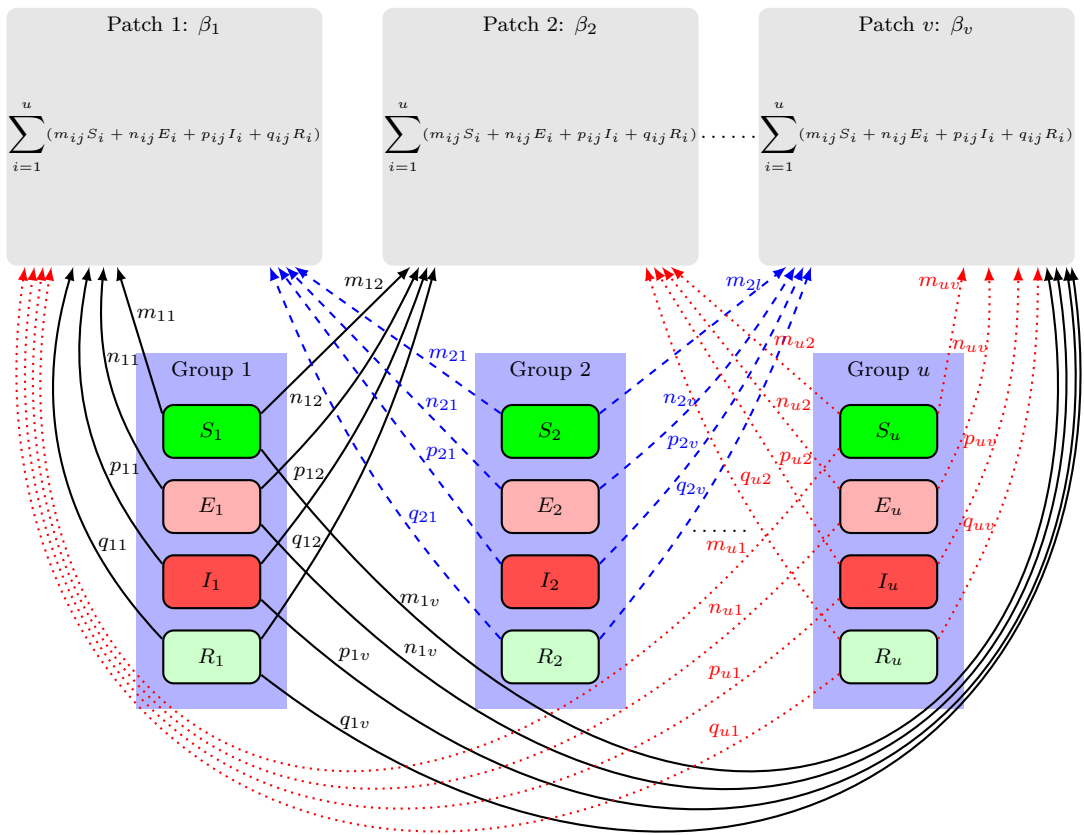}
\caption{Flow diagram of Model \ref{PatchGen}.}
\label{fig:FlowEnv} 
\end{figure}

Model (\ref{PatchGen}) could be written in the compact form,
{\small{
\begin{equation} \label{PatchGenMatrix}
\left\{\begin{array}{llll}
\dot{\mathbf{S}}= {\mathbf{\Lambda}}-\textrm{diag}(\mathbf{S})\mathbb M\textrm{diag}(\mathcal{B})\textrm{diag}^{-1}(\mathbb{M}^{T}\mathbf{S}+  \mathbb{N}^{T}\mathbf{E}+\mathbb{P}^{T}\mathbf{I}+ \mathbb{Q}^{T}\mathbf{R} )\mathbb{P}^{T}\mathbf{I}  -\textrm{diag}(\mu) \mathbf{S}+\textrm{diag}(\eta)\mathbf{R}\\
\dot{\mathbf{E}}= \textrm{diag}(\mathbf{S})\mathbb M\textrm{diag}(\mathcal{B})\textrm{diag}^{-1}(\mathbb{M}^{T}\mathbf{S}+  \mathbb{N}^{T}\mathbf{E}+\mathbb{P}^{T}\mathbf{I}+ \mathbb{Q}^{T}\mathbf{R} )\mathbb{P}^{T}\mathbf{I}         -\textrm{diag}(\nu+\mu)\mathbf{E}\\
\dot{\mathbf{I}}=\textrm{diag}(\nu)\mathbf{E}- \textrm{diag}(\gamma+\mu+\delta)\mathbf{I} \\
\dot{\mathbf{R}}=\textrm{diag}(\gamma)\mathbf{I}- \textrm{diag}(\eta+\mu)\mathbf{R} 
\end{array}\right.
\end{equation}
}}
where $\mathbf{S}=[S_1,S_2,\dots,S_u]^{T}$, $\mathbf{E}=[E_1,E_2,\dots,E_u]^{T}$, $\mathbf{I}=[I_1,I_2,\dots,I_u]^{T}$ and $\mathbf{R}=[R_1,R_2,\dots,R_u]^{T}$. The matrices $\mathbb{M}=(m_{ij})_{\substack{
      1 \leq i \leq u, \\
      1 \leq j \leq v}}$, $\mathbb N=(n_{ij})_{\substack{
      1 \leq i \leq u, \\
      1 \leq j \leq v}}$, $\mathbb P=(p_{ij})_{\substack{
      1 \leq i \leq u, \\
      1 \leq j \leq v}}$ and $\mathbb Q=(q_{ij})_{\substack{
      1 \leq i \leq u, \\
      1 \leq j \leq v}}$ represent the residence time matrices of susceptible, latent, infectious and recovered individuals respectively.  Moreover, ${\mathbf{\Lambda}}=[\Lambda_1,\Lambda_2,\dots,\Lambda_u]^{T}$, $\mathcal B=[\beta_1,\beta_2,\dots,\beta_v]^{T}$,  $\mu =[\mu_1,\mu _2,\dots,\mu _u]^{T}$, $\nu =[\nu_1,\nu _2,\dots,\nu _u]^{T}$, $\gamma =[\gamma_1,\gamma _2,\dots,\gamma _u]^{T}$, $\delta =[\delta_1,\delta _2,\dots,\delta _u]^{T}$  and $\eta =[\eta_1,\eta _2,\dots,\eta _u]^{T}$.\\  
      
\n Model (\ref{PatchGenMatrix}) brings added value to the existing literature in the following ways:
\begin{enumerate}
\item The structure of the host population is different and independent from the patches where the infection takes place. Indeed, in the previous epidemic models describing human dispersal or mixing (Eulerian or Lagrangian), hosts' structure unit and the geographical landscape unit, be it group or patch, is the same and  homogeneous in term of \textit{transmission rate}. Our model captures added heterogeneity in the sense that we decouple the structure of the host to that of patches. For instance, our framework fits well for nosocomial diseases (hospital-acquired infections), where the  hospitals  could  be treated as patches and host's groups as gender or age (see \cite{eckenrode2014association,kaplan2002hospitalized} for the effects of gender and age on nosocomial infections).
\item In our formulation, there is no need to measure contacts rates, a difficult task for nearly all diseases that are not either sexually transmitted or vector-borne. Each patch is defined by its specific risk of infection that could be tied to environmental or hygienic conditions. Hence, susceptibility is not individual-based nor group-based as in classical formulation of multi-group models (the contact matrices in these type of models are known as WAIFW, i.e., Who Acquires Infection From Whom \cite{AndMay91}), but a patch specific risk. In fact, our framework is capable of capturing a wide-range of modeling scenarios, including group-susceptibility.  Indeed, if $g_i$ is the risk of infection of Group $i$, \mbox{$i=1,2,\dots,u$}, it suffices to replace $S_i$ by $g_iS_i$ in only the infection terms in (\ref{PatchGen}). That is, the dynamics of susceptible and latent hosts, for $i=1,2,\dots,u$ will be:
$$\dot S_{i}= \Lambda_i-\sum_{j=1}^v\beta_jm_{ij}g_iS_i\dfrac{\sum_{k=1}^up_{kj}I_k}{\sum_{k=1}^u(m_{kj}S_k+n_{kj}E_k+p_{kj}I_k+q_{kj}R_k)}-\mu_iS_i+\eta_iR_i,$$
and 
$$\dot E_i=\sum_{j=1}^v\beta_jm_{ij}g_iS_i\dfrac{\sum_{k=1}^up_{kj}I_k}{\sum_{k=1}^u(m_{kj}S_k+n_{kj}E_k+p_{kj}I_k+q_{kj}R_k)}-(\nu_i+\mu_i)E_i.$$
For the sake of simplicity, we considered the case where all host groups have the same risk of infection, though all the results obtained in this paper hold without this simplification.

The risk in each patch may be fixed, as in Model (\ref{PatchGenMatrix}), or variable and dependent of the \textit{effective} patch population (See Subsection \ref{DDTransmission}).  The prospect of infection is tied to the environmental risk and time spent in that environment. This fits, for example, pandemic influenza in schools and, again, the nosocomial infections (length of stay in hospitals and their corresponding risks). 
These residences times and patch related risks are easier to quantify than \textit{contact rates}. This paper extend earlier results in \cite{BicharaCCC2015,bichara2015sis}.
\item The model allows individuals of different groups to move across patches without losing their identities. This approach allows a more targeted control strategy for public health benefit. Therefore, the model follows a Lagrangian approach and generalize \cite{BicharaCCC2015,bichara2015sis,cosner2009effects,iggidr2014dynamics,rodriguez2001models,Ruktanonchai2016}.

\item There are different mobility patterns depending on the epidemiological class of individuals. This allows us to highlight and assess the effects of hosts' behavior through social distancing and their predilection for specific patches on the disease dynamics. Although the differential mobility have been considered in an Eulerian setting \cite{SalVdd06,xiao2014transmission}, its incorporation in a Lagrangian setting is new and is an extension of \cite{BicharaCCC2015,bichara2015sis,cosner2009effects,falcon2016day,iggidr2014dynamics,rodriguez2001models,Ruktanonchai2016} (for which mobility is independent of hosts' epidemiological class).
\item In this framework, we consider only patches where the infection takes place (hospitals, schools, malls, etc) whereas previous models suppose that the patches are distributed over the whole space. In short, the mobility matrices are not assumed to be stochastic.In this case, a natural condition on the mobility matrices arises: $\mathbb X\mathbf{1}\leq\mathbf{1}$, for $\mathbb X\in\{\mathbb M, \mathbb N, \mathbb P, \mathbb Q\}$, where $\mathbf{1}$ is the vector whose components are all equal to unity. These conditions stem from the fact that the added proportion of time spend in all patches cannot be more that $100\%$. However, as pointed out by a reviewer, the stochasticity of the mobility matrices is not really restrictive. Indeed, as we are considering an arbitrary number of patches, we can, without loss of generality, add an additional patch within which individuals spent ``the rest of their time" and  where no infection takes place in it. That is, $\beta_{v+1}=0$.
\end{enumerate} 

We denote by $\mathbf{N}$ the vector of populations of each group. The dynamics of the population in each group is given by the following: 

$$\dot{\mathbf{N}}=\mathbf{\Lambda}-\mu\circ \mathbf{N}-\delta\circ\mathbf{I}\leq \mathbf{\Lambda}-\mu\circ \mathbf{N}$$

where $\circ$ denotes the Hadamard product. Thus,  the set defined by 
\[\Omega=\left \{ (\mathbf{S}, \mathbf{E}, \mathbf{I},\mathbf{R})\in\R^{4u}_+\; \mbox{ \LARGE $\mid$  } \;\mathbf{S}+\mathbf{E}+\mathbf{I}+\mathbf{R}\leq\mathbf{\Lambda}\circ\dfrac{1}{\mu}\right \} \] 
is a compact attracting  positively invariant for System (\ref{PatchGenMatrix}). \\

\n The disease-free equilibrium (DFE) of System (\ref{PatchGenMatrix}) is given by $(\mathbf{S}^\ast, \mathbf{0}, \mathbf{0},\mathbf{0})$ where $\mathbf{S}^\ast=\mathbf{\Lambda}\circ\dfrac{1}{\mu}$.

\begin{remark}
If the susceptible or infected individuals do not go to the patches where the infection takes place, either due to  intervention strategy or social distancing, that is when the residence time matrices $\mathbb M$ or $\mathbb P$ are the null matrix (the susceptible individuals do not spend any time in the considered patches), the disease does not spread and eventually dies out. 

\end{remark}
We compute the basic reproduction number following   \cite{MR1057044,VddWat02}. By decomposing the infected compartments of (\ref{PatchGenMatrix}) as a sum of new infection terms and transition terms,

\begin{eqnarray} \label{R01}
\left(\begin{array}{c}
\dot{\mathbf{E}}\\
\dot{\mathbf{I}}
\end{array}\right)&=&\mathcal F(\mathbf{E},\mathbf{I})+\mathcal V(\mathbf{E},\mathbf{I})\nonumber\\
&=&\left(\begin{array}{c}
\textrm{diag}(\mathbf{S})\mathbb M\textrm{diag}(\mathcal{B})\textrm{diag}^{-1}(\mathbb{M}^{T}\mathbf{S}+  \mathbb{N}^{T}\mathbf{E}+\mathbb{P}^{T}\mathbf{I}+ \mathbb{Q}^{T}\mathbf{R} )\mathbb{P}^{T}\mathbf{I}\\
0
\end{array}\right)\nonumber\\
& &+
\left(\begin{array}{c}
       -\textrm{diag}(\nu+\mu)\mathbf{E}\\
\textrm{diag}(\nu)\mathbf{E}- \textrm{diag}(\gamma+\mu+\delta)\mathbf{I}
\end{array}\right)\nonumber
\end{eqnarray}

The Jacobian matrix at the DFE of $\mathcal F(\mathbf{E},\mathbf{I})$ and $\mathcal V(\mathbf{E},\mathbf{I})$ are given by:

$$F=D\mathcal F(\mathbf{E},\mathbf{I})\Bigg|_\textrm{DFE}=\left(\begin{array}{cccc}
\textbf{0}_{u,u} & \textrm{diag}(\mathbf{S^\ast})\mathbb M\textrm{diag}(\mathcal{B})\textrm{diag}^{-1}(\mathbb{M}^{T}\mathbf{S^\ast} )\mathbb{P}^{T}\\
\textbf{0}_{u,u}& \textbf{0}_{u,u}
\end{array}\right)$$
and,
$$V=D\mathcal V(\mathbf{E},\mathbf{I})\Bigg|_\textrm{DFE}=\left(\begin{array}{cccc}
-\textrm{diag}( \mu+\nu)& \textbf{0}_{u,u}\\ 
\textrm{diag}(\nu)  & -\textrm{diag}(\mu+\gamma+\delta) 
\end{array}\right)$$

Hence, we obtain $$-V^{-1}=\left(\begin{array}{cccc}
\textrm{diag}^{-1}( \mu+\nu)& \textbf{0}_{u,u}\\ 
\textrm{diag}(\nu) \textrm{diag}^{-1}((\mu+\nu)\circ (\mu+\gamma+\delta)) & \textrm{diag}^{-1}(\mu+\gamma+\delta)
\end{array}\right)$$
The basic reproduction number is the spectral radius of the next generation matrix

$$
-FV^{-1}=\left(\begin{array}{cccc}
Z \textrm{diag}(\nu) \textrm{diag}^{-1}((\mu+\nu)\circ (\mu+\gamma+\delta))& Z\textrm{diag}^{-1}(\mu+\gamma+\delta)\\
\textbf{0}_{u,u}& \textbf{0}_{u,u}
\end{array}\right)
$$
where 
$$Z=\textrm{diag}(\mathbf{S^\ast})\mathbb M\textrm{diag}(\mathcal{B})\textrm{diag}^{-1}(\mathbb{M}^{T}\mathbf{S^\ast} )\mathbb{P}^{T}$$
 Finally, the basic reproduction number for $u$ groups and $v$ patches is given by
 
 $$\mathcal R_0(u,v)=\rho(Z \textrm{diag}(\nu) \textrm{diag}^{-1}((\mu+\nu)\circ (\mu+\gamma+\delta)))$$

The disease-free equilibrium is asymptotically stable whenever $\mathcal R_0(u,v)<1$ and unstable if $\mathcal R_0(u,v)>1$ \cite{MR1057044,VddWat02}.
\section{Effects of heterogeneity on the basic reproduction number}\label{sec:Heterogeneity}
In this section, we investigate the effects of \textit{patchiness}, \textit{groupness} and mobility on the basic reproduction number. More particularly, how the basic reproduction number changes its monotonicity with respect to the number of patches, groups and mobility patterns of individuals.\\

The following theorem gives the monotonicity of the basic reproduction with respect the residence times patterns of the infected individuals. 

\begin{theorem}\label{TheoremInegaR0P}\hfill

The basic reproduction number $\mathcal R_0(u,v)$ is a nondecreasing function with respect to $\mathbb P$, that is, the infected individuals movement patterns. 
\end{theorem}
\begin{proof}\hfill

\n Recall that $\mathcal R_0(u,v)=\rho(Z \textrm{diag}(\nu) \textrm{diag}^{-1}((\mu+\nu)\circ (\mu+\gamma+\delta)))$ where,\\  $Z=\textrm{diag}(\mathbf{S^\ast})\mathbb M\textrm{diag}(\mathcal{B})\textrm{diag}^{-1}(\mathbb{M}^{T}\mathbf{S^\ast} )\mathbb{P}^{T}$. The matrix $Z$ is linear in $\mathbb P$ and  has all non-negative entries. We consider the  order relation for the matrices as follows: $A\leq B$ if $a_{ij}\leq b_{ij}$, for all $i$ and all $j$, where $a_{ij}$ and $b_{ij}$ are entries of $A$ and $B$ respectively. Also, $A<B$ if $A\leq B$ and $A\neq B$.  Hence, since the Perron-Frobenius theorem \cite{0815.15016} (Corollary 1.5, page 27) guarantees that for any positives matrices $A$ and $B$ such that $A\geq B\geq0$, then $\rho(A)\geq\rho(B)$, we deduce that, for any matrix $\mathbb{P}'\geq\mathbb{P}$,
\begin{eqnarray}
\mathcal{R}_0(u,v,\mathbb{P})&=&\rho(\textrm{diag}(\mathbf{S^\ast})\mathbb M\textrm{diag}^{-1}(\mathcal{B})\textrm{diag}(\mathbb{M}^{T}\mathbf{S^\ast} )\mathbb{P}^{T}\textrm{diag}(\nu) \textrm{diag}^{-1}((\mu+\nu)\circ (\mu+\gamma+\delta)))\nonumber\\
&\leq&\rho(\textrm{diag}(\mathbf{S^\ast})\mathbb M\textrm{diag}(\mathcal{B})\textrm{diag}^{-1}(\mathbb{M}^{T}\mathbf{S^\ast} )\mathbb{P}'^{T}\textrm{diag}(\nu) \textrm{diag}^{-1}((\mu+\nu)\circ (\mu+\gamma+\delta)))\nonumber\\
&:=&\mathcal{R}_0(u,v,\mathbb{P}')\nonumber
\end{eqnarray}
\end{proof}
The variation in monotonicity of $\mathcal R_0(u,v)$ with respect to the residence times patterns of susceptible individuals, that is $\mathbb M$,  is more complicated and difficult to assess in general and even in some more restrictive particular cases (see Remark \ref{MRank1Independance}). \\

\n Hereafter, we define two bounding quantities tied to the global basic reproduction number:

\begin{eqnarray}
\tilde{\mathcal R}_0^i(u,v)&=&\frac{\nu_i}{(\nu_i+\mu_i)(\gamma_i+\mu_i+\delta_i)}\sum_{j=1}^v\frac{\beta_jm_{ij}S_i^\ast p_{ij}}{\sum_{k=1}^um_{kj}S_k^\ast}\nonumber\\
& = & \frac{\beta_i\nu_i}{(\nu_i+\mu_i)(\gamma_i+\mu_i+\delta_i)}\sum_{j=1}^v\left( \frac{\beta_j}{\beta_i} \right)\frac{m_{ij}S_i^\ast p_{ij}}{\sum_{k=1}^um_{kj}S_k^\ast},\nonumber
\end{eqnarray}

and,
$$\mathcal R_0^{i}=\frac{\nu_i}{(\mu_i+\nu_i)(\mu_i+\gamma_i+\delta_i)}\sum_{k=1}^v\beta_kp_{ik}$$
It is worthwhile noting that $\mathcal R_0^{i}=\mathcal R_0 (1,v)$. That is, $\mathcal R_0^{i}$ is also the basic reproduction number of the global system in presence of one group only, namely the $i^{th}$, spread over $v$ patches. $\mathcal R_0^{i}$ could be seen as a group specific ``reproduction number".\\
The quantity $\tilde{\mathcal R}_0^i(u,v)$ could be heuristically seen as the sum of the average number of cases produced by an infected of group $i$ over all patches, in presence of other groups.

In the following theorem, we explore how the general basic reproduction number  $\mathcal R_0 (u,v)$ is tied to these specific reproduction numbers and whether it increases or decreases when the number of patches and/or groups changes.  An underlying assumption in the following theorem is that when adding patches, the proportion of time spent in the existing patches remain exactly the same.\\

\begin{theorem}\label{R0Bounds}\hfill

We have the following inequalities:
\begin{enumerate}
\item $ \displaystyle \max\left\{\max_{i=1,\dots,u}\tilde{\mathcal R}_0^i(u,v),\min_{i=1,\dots,u} \mathcal R_0^{i}\right\}\leq\mathcal R_0(u,v)\leq \max_{i=1,\dots,u} \mathcal R_0^{i}$
\item $\mathcal R_0(u,v)\geq\mathcal R_0(1,v)\geq\mathcal R_0(1,1).$
\item For a fixed number of groups $u$, $\mathcal R_0(u,v)\geq\mathcal R_0(u,v')$ where $v$ and $v'$ are integers such that $v\geq v'$. \label{thm:Ineq3}
\end{enumerate}
\end{theorem}
\begin{proof}\hfill

\noindent \textit{1.} We prove first that $\mathcal R_0(u,v)\geq \displaystyle\max_{i=1,\dots,u}\tilde{\mathcal R}_0^i(u,v)$ and then $\displaystyle \min_{i=1,\dots,n}\mathcal R_0^{i} \leq \mathcal R_0(u,v)\leq \max_{i=1,\dots,n}\mathcal R_0^{i}$.

 Let $e_i$ the $i-$th vector of the canonical basis of $\mathbb R^{4u}$. We have
$$e_i^T\diag(\mathbf{S^\ast})\mathbb M=(m_{i1}S^\ast_i,m_{i2}S^\ast_i,\dots,m_{iv}S^\ast_i)$$
It follows that, 
$$e_i^T\diag(\mathbf{S^\ast})\mathbb M\textrm{diag}(\mathcal{B})=(\beta_1m_{i1}S^\ast_i,\beta_2m_{i2}S^\ast_i,\dots,\beta_vm_{iv}S^\ast_i)$$
We also have
$$\mathbb M^{T}\mathbf{S}^\ast=\left(\begin{array}{c}
  \sum_{k=1}^um_{k1}S_k^\ast\\
 \sum_{k=1}^um_{k2}S_k^\ast\\
 \vdots\\
  \sum_{k=1}^um_{kv}S_k^\ast\
\end{array}\right)
$$
Since $\mathbb P^{T}e_i$ is the $i-$th column of $\mathbb P^{T}$, we obtain:

$$\textrm{diag}^{-1}(\mathbb M^{T}\mathbf{S}^\ast)\mathbb P^{T}e_i=\left(\begin{array}{c}
 \frac{p_{i1}}{ \sum_{k=1}^um_{k1}S_k^\ast}\\
 \frac{p_{i2}}{ \sum_{k=1}^um_{k2}S_k^\ast}\\
 \vdots\\
 \frac{p_{iv}}{  \sum_{k=1}^um_{kv}S_k^\ast}
\end{array}\right)
$$
Hence, the diagonal elements of $\mathbb M\textrm{diag}(\mathcal{B})\textrm{diag}(\mathbb{M}^{T}\mathbf{S^\ast} )^{-1}\mathbb{P}^{T}$ is given by

\begin{eqnarray}
e_i^{T}\diag(\mathbf{S^\ast})\mathbb M\textrm{diag}(\mathcal{B})\textrm{diag}^{-1}(\mathbb{M}^{T}\mathbf{S^\ast} )\mathbb{P}^{T}e_i&=&\frac{\beta_1m_{i1}p_{i1}S^\ast_i}{ \sum_{k=1}^um_{k1}S_k^\ast}+
 \frac{\beta_2m_{i2}p_{i2}S^\ast_i}{ \sum_{k=1}^um_{k2}S_k^\ast}+
 \cdots+
 \frac{\beta_vm_{iv}p_{iv}S^\ast_i}{  \sum_{k=1}^um_{kv}S_k^\ast}\nonumber\\
&=&\sum_{j=1}^v \frac{\beta_jm_{ij}p_{ij}S^\ast_i}{ \sum_{k=1}^um_{kj}S_k^\ast}\nonumber
\end{eqnarray}
This implies that, for all $i=1,\cdots,v$, $\tilde{\mathcal R}_0^i(u,v)$ is a diagonal element of the next generation matrix. Since the spectral radius of a matrix is the greater or equal to its diagonal elements, we can conclude that
$\mathcal R_0(u,v)\geq \tilde{\mathcal R}_0^i$ for all  $i=1,\cdots,u$. This implies that \

\begin{equation}\label{RTildeIneq}
\mathcal R_0(u,v)\geq \displaystyle\max_{i=1,\dots,u}\tilde{\mathcal R}_0^i(u,v)
\end{equation}

It remains to prove that $ \displaystyle \min_{i=1,\dots,u} \mathcal R_0^{i}\leq\mathcal R_0(u,v)\leq \max_{i=1,\dots,u} \mathcal R_0^{i}$. The basic reproduction number is given by $\mathcal R_0(u,v)=\rho(Z \textrm{diag}(\nu) \textrm{diag}^{-1}((\mu+\nu)\circ (\mu+\gamma+\delta)))$ where $$Z=\textrm{diag}(\mathbf{S^\ast})\mathbb M\textrm{diag}(\mathcal{B})\textrm{diag}^{-1}(\mathbb{M}^{T}\mathbf{S^\ast} )\mathbb{P}^{T}$$ It can be shown that the elements of this matrix are the following:

\begin{equation}\label{R0Elements} z_{ij}=\frac{\nu_j}{(\mu_j+\nu_j)(\mu_j+\gamma_j+\delta_j)}\sum_{k=1}^v\frac{\beta_km_{ik}p_{jk}S_i^\ast}{\sum_{l=1}^um_{lk}S_l^\ast} \quad \forall \quad 1\leq i,j\leq u.
\end{equation}
If  $\mathbb M\mathbb P^T$ is irreducible, the matrix $Z \textrm{diag}(\nu) \textrm{diag}^{-1}((\mu+\nu)\circ (\mu+\gamma+\delta)))$ is irreducible, and therefore its spectral radius satisfy
the Frobenius' inequality (\cite{horn1985matrix}, Theorem 8.1.22, page 492):
$$\min_jz_{j}\leq \mathcal R_0(u,v)\leq\max_j z_{j}$$
where $z_{j}=\sum_{i=1}^{u} z_{ij}$ and $z_{ij}$ are given by (\ref{R0Elements}). We have:
\begin{eqnarray}
z_{j} &=&\sum_{i=1}^{u}z_{ij}\nonumber\\
&=&\sum_{i=1}^{u}\frac{\nu_j}{(\mu_j+\nu_j)(\mu_j+\gamma_j+\delta_j)}\sum_{k=1}^v\frac{\beta_km_{ik}p_{jk}S_i^\ast}{\sum_{l=1}^um_{lk}S_l^\ast}\nonumber\\
&=&\frac{\nu_j}{(\mu_j+\nu_j)(\mu_j+\gamma_j+\delta_j)}\sum_{i=1}^{u}\sum_{k=1}^v\frac{\beta_km_{ik}p_{jk}S_i^\ast}{\sum_{l=1}^um_{lk}S_l^\ast}\nonumber\\
&=&\frac{\nu_j}{(\mu_j+\nu_j)(\mu_j+\gamma_j+\delta_j)}\sum_{k=1}^v\sum_{i=1}^{u}\frac{\beta_km_{ik}p_{jk}S_i^\ast}{\sum_{l=1}^um_{lk}S_l^\ast}\nonumber\\
&=&\frac{\nu_j}{(\mu_j+\nu_j)(\mu_j+\gamma_j+\delta_j)}\sum_{k=1}^v\frac{\beta_kp_{jk}}{\sum_{l=1}^um_{lk}S_l^\ast}\sum_{i=1}^{u}m_{ik}S_i^\ast\nonumber\\
&=&\frac{\nu_j}{(\mu_j+\nu_j)(\mu_j+\gamma_j+\delta_j)}\sum_{k=1}^v\beta_kp_{jk}\nonumber\\
&:=& \mathcal R_0^j \nonumber
\end{eqnarray}
 Hence, \begin{equation}\label{RiIneq}
 \min_i  \mathcal R_0^i\leq \mathcal R_0(u,v)\leq\max_i  \mathcal R_0^i 
 \end{equation}
 The relations (\ref{RTildeIneq}) and (\ref{RiIneq}) imply the desired inequality. \\
 
 \noindent \textit{2.} By using the inequality proved in the first part, we have:
 \begin{eqnarray}
 \mathcal R_0(u,v) &\geq&\min_{i=1,\dots,u}  \mathcal R_0^i\nonumber\\
 &:=& \mathcal{R}_0(1,v),\nonumber
  \end{eqnarray} 
  Finally, we have: 
   \begin{eqnarray}
 \mathcal R_0(1,v) &=&\mathcal R_0^{1}\nonumber\\
  &=&\frac{\nu_1}{(\mu_1+\nu_1)(\mu_1+\gamma_1+\delta_1)}\sum_{k=1}^v\beta_kp_{1k}\nonumber\\
  &\geq&\frac{\beta_{1}p_{11}\nu_1}{(\mu_1+\nu_1)(\mu_1+\gamma_1+\delta_1)}\nonumber\\
 &:=& \mathcal{R}_0(1,1)\nonumber
  \end{eqnarray} 
  3. Let $u$ a fixed number of groups. We would like to prove that $\mathcal R_0(u,v)\geq\mathcal R_0(u,v')$ for any $v\geq v'$. Since, $\mathcal R_0(u,v)=\rho(Z \textrm{diag}(\nu) \textrm{diag}^{-1}((\mu+\nu)\circ (\mu+\gamma+\delta)))$ and the number of groups is fixed, the epidemiological parameters remain the same for any number of patches. Hence, it remains to compare $Z_{v}$ and $Z_{v'}$ where $Z$ is the part of the next generation matrix that depends on the number of patches.\\
  For $v$ patches, we have $$Z_{v}^{ij}=\sum_{k=1}^v\displaystyle\frac{\beta_km_{ik}p_{jk}S_i^\ast}{\sum_{l=1}^{u}m_{lk}S_l^\ast}$$
  For $v'$ patches, $$Z_{v'}^{ij}=\sum_{k=1}^{v'}\displaystyle\frac{\beta_km_{ik}p_{jk}S_i^\ast}{\sum_{l=1}^{u}m_{lk}S_l^\ast}$$
  Hence, for $v\geq v'$, we have clearly $Z_{v}^{ij}\geq Z_{v'}^{ij}$. Hence, thanks to Perron-Frobebenius' theorem, we conclude that $\mathcal R_0(u,v)\geq\mathcal R_0(u,v')$.  
\end{proof}
\begin{remark}\hfill
\begin{itemize}
\item The inequality in Item~\ref{thm:Ineq3} of Theorem \ref{R0Bounds}
 is independent of the risk of infection in the additional patches. 
 \item If the residence times network configuration changes due the newly added patches, the increasing property of the basic reproduction number with respect to the number of patches (Item~\ref{thm:Ineq3} of Theorem \ref{R0Bounds}) may not hold. This is an interesting avenue to exploring the monotonicity of $\mathcal{R}_0$ and/or the dynamics of the disease.
 \end{itemize}
\end{remark}
\n We investigate relevant modeling scenarios where the expression of the general basic reproduction number for $u$ patches and $v$ patches, $\mathcal R_0(u,v)$, could be explicitly obtained. In the rest of the paper, we use $\langle x \, \mbox{\Large $\mid$} \, y  \rangle$ to denote the canonical scalar product.\\

\begin{theorem}\label{Mrank1}\hfill\

If the susceptible residence times matrix  $\mathbb M$ is of rank one, an explicit expression of $\mathcal R_0$ is given by
\begin{eqnarray*}
\mathcal R_0(u,v)&=&\left(  \xi^T\mathbf{S}^\ast\right)^{-1}\mathcal B^T\mathbb P^T\emph{\textrm{diag}}^{-1}(\nu)\emph{\textrm{diag}}((\mu+\nu)\circ(\mu+\gamma+\delta))\emph{\textrm{diag}}(\mathbf{S^\ast})\xi
\\
&:=&
\left(  \xi^T\mathbf{S}^\ast\right)^{-1}\Biggl\langle
\mathcal B \; \mbox{ \LARGE $\mid$  } \;  \mathbb P^T\emph{\textrm{diag}}(\nu)\emph{\textrm{diag}}^{-1}((\mu+\nu)\circ(\mu+\gamma+\delta))\emph{\textrm{diag}}(\mathbf{S^\ast})\xi
\Biggr\rangle
\end{eqnarray*}
where $\xi\in\R^u$ is such that $\mathbb M=\xi^Tm$, with $m\in\R^v$. Moreover, if the matrix $\mathbb M$ is stochastic, we have:
$$\mathcal R_0(u,v)=
\left(  \mathbf{1}^{T}\mathbf{S}^\ast\right)^{-1}\Biggl\langle
\mathcal B \; \mbox{ \LARGE $\mid$  } \;  \mathbb P^T\emph{\textrm{diag}}(\nu)\emph{\textrm{diag}}^{-1}((\mu+\nu)\circ(\mu+\gamma+\delta))\mathbf{S^\ast}
\Biggr\rangle
$$

\end{theorem}

\begin{proof}\hfill

If the susceptible residence times matrix  $\mathbb M$ is of rank one, it exist a vector $\xi\in\R^u$ and a vector $m\in\R^v$ such that $\mathbb M=\xi m^{T}$. We have the following:
$$\mathbb{M}^{T}\mathbf{S}^\ast=m\xi^{T}\mathbf{S}^\ast=\langle \xi  \; \mid   \; \mathbf{S}^\ast\rangle m$$
Hence, 

$$\diag^{-1}(\mathbb M^{T}\mathbf{S}^\ast)=\diag^{-1}(\langle \xi  \; \mid   \; \mathbf{S}^\ast\rangle m)=\langle \xi  \; \mid   \; \mathbf{S}^\ast\rangle^{-1}\diag^{-1}(m)$$
and

\begin{eqnarray}
Z&=&\textrm{diag}(\mathbf{S^\ast})\mathbb M\textrm{diag}(\mathcal{B})\textrm{diag}^{-1}(\mathbb{M}^{T}\mathbf{S^\ast} )\mathbb{P}^{T}\nonumber\\
&=&\textrm{diag}(\mathbf{S^\ast})\xi m^T\textrm{diag}(\mathcal{B})\langle \xi  \; \mid   \; \mathbf{S}^\ast\rangle^{-1}\diag^{-1}(m)\mathbb{P}^{T}\nonumber\\
&=&\langle \xi  \; \mid   \; \mathbf{S}^\ast\rangle^{-1}\textrm{diag}(\mathbf{S^\ast})\xi m^T\textrm{diag}(\mathcal{B})\diag^{-1}(m)\mathbb{P}^{T}\nonumber\\
&=&\langle \xi  \; \mid   \; \mathbf{S}^\ast\rangle^{-1}\textrm{diag}(\mathbf{S^\ast})\xi m^T\diag^{-1}(m)\textrm{diag}(\mathcal{B})\mathbb{P}^{T}\nonumber\\
&=&\langle \xi  \; \mid   \; \mathbf{S}^\ast\rangle^{-1}\textrm{diag}(\mathbf{S^\ast}) \xi  \mathbf{1}^T  \textrm{diag}(\mathcal{B})\mathbb{P}^{T}\quad\textrm{because}\quad m^T\diag^{-1}(m)=\mathbf{1}^T\nonumber\\
&=&\langle \xi  \; \mid   \; \mathbf{S}^\ast\rangle^{-1}\textrm{diag}(\mathbf{S^\ast})\xi\mathcal{B}^T\mathbb{P}^{T}
\end{eqnarray}
We deduce that the non-zero diagonal block of the next generation matrix could be written as:
$$Z\textrm{diag}(\nu) \textrm{diag}^{-1}((\mu+\nu)\circ (\mu+\gamma+\delta)))= \langle \xi  \; \mid   \; \mathbf{S}^\ast\rangle^{-1}\textrm{diag}(\mathbf{S^\ast})\xi\mathcal{B}^T\mathbb{P}^{T}\textrm{diag}(\nu) \textrm{diag}^{-1}((\mu+\nu)\circ (\mu+\gamma+\delta)))$$
This matrix is clearly of rank 1, since it could be written as $wz^T$ where $w\in\R^u$ and $w\in\R^v$. Hence, its unique non zero eigenvalue is 
$$\mathcal R_0(u,v)=\langle \xi  \; \mid   \; \mathbf{S}^\ast\rangle^{-1}\mathcal{B}^T\mathbb{P}^{T}\textrm{diag}(\nu) \textrm{diag}^{-1}((\mu+\nu)\circ (\mu+\gamma+\delta))) \textrm{diag}(\mathbf{S^\ast})\xi$$
or, equivalently,
$$\mathcal R_0(u,v)=
\left(  \xi^T\mathbf{S}^\ast\right)^{-1}\Biggl\langle
\mathcal B \; \mbox{ \LARGE $\mid$  } \;  \mathbb P^T\textrm{diag}(\nu)\textrm{diag}^{-1}((\mu+\nu)\circ(\mu+\gamma+\delta)) \textrm{diag}(\mathbf{S^\ast})\xi
\Biggr\rangle
$$
Now, if $\mathbb M$ is of rank one and stochastic, that is , $\sum_{j=1}^vm_{ij}=1$, for all $i=1,\dots,u$, it is not difficult to show that $\xi=\mathbf{1}$, where $\mathbf{1}$ is the vector whose components are all equal to unity. This leads to
$$\mathcal R_0(u,v)=
\left(  \mathbf{1}^{T}\mathbf{S}^\ast\right)^{-1}\Biggl\langle
\mathcal B \; \mbox{ \LARGE $\mid$  } \;  \mathbb P^T\textrm{diag}(\nu)\textrm{diag}^{-1}((\mu+\nu)\circ(\mu+\gamma+\delta))\mathbf{S^\ast}
\Biggr\rangle
$$
\end{proof}

\begin{remark}\label{MRank1Independance}\hfill

 If the residence times matrix of susceptible individuals, that is $\mathbb M$, is of rank one and stochastic, the basic reproduction number is independent of $\mathbb M$.
\end{remark}
It is worthwhile noting that there is a special case for which the result of Remark \ref{MRank1Independance} holds even if the matrix $\mathbb M$ is not stochastic
 but only of rank one and sub-stochastic. 
 Indeed, by adding a new patch $v+1$ with $\beta_{v+1}=0$ where the hosts of different groups spend ``the rest of their times", the new mobility matrices will become the matrices $\tilde{\mathbb M}=(\mathbb M,\mathbb M'), \tilde{\mathbb N}=(\mathbb N,\mathbb N'), \tilde{\mathbb P}=(\mathbb P,\mathbb P')$ and $\tilde{\mathbb Q}=(\mathbb Q,\mathbb Q')$, where $\mathbb M', \mathbb N', \mathbb P'$ and $\mathbb Q'$  are column vectors. The new mobility matrices are stochastic and $\mathcal R_0(u,v,\mathbb M,\mathbb P)=\mathcal R_0(u,v+1,\tilde{\mathbb M},\tilde{\mathbb P})$ since $\beta_{v+1}=0$. Hence, if $\mathbb M$ and $\tilde{\mathbb M}$ are of rank one, the basic reproduction number is still independent of $\mathbb M$. In this case,  the matrix $\mathbb M$ could be expressed as  $\mathbf{1}m^{T}$ with $\sum_{j=1}^vm_{j}<1$. 
 Thus, there is a special case when $\mathbb M$ is rank 1, yet sub-stochastic, and the reproduction number does not depend on $\mathbb M$.\\

 From a modeling standpoint, the rank one condition of $\mathbb M$ (i.e., $\mathbb M=\xi m^{T}$ with $\xi\in\R^u$ and $m\in\R^v$) can be interpreted as follows:\\
$\bullet$ The ratio of the proportions of time spent in any given patch by susceptible individuals belonging to two different groups, is identical. Indeed, for any given group $i$, the ratio of the proportion of time spent in any given patch by susceptible individual is:
 \begin{eqnarray} \frac{m_{ij}}{\dsum_{k=1}^vm_{ik}}&=&\frac{\xi_{i}m_j}{\dsum_{k=1}^v\xi_{i}m_k}\nonumber\\
 &=&\dfrac{m_j}{\dsum_{k=1}^vm_k},\nonumber\end{eqnarray} which is independent of $i$.  
Moreover, if $\mathbb M$ is stochastic, we deduce that the susceptible of each group spend the exact proportion of time in any given patch, since $\sum_{k=1}^vm_k=1$.\\
$\bullet$  A straightforward case that stems from the previous point is whenever there is one patch and multiple groups; or when there are multiple patches and one group.\\

Similar remarks hold when the matrix $\mathbb P$ is of rank one, which is dealt in the next theorem. 
\begin{theorem}\label{Prank1}\hfill\

If the infected residence times matrix  $\mathbb P$ is of rank one, an explicit expression of $\mathcal R_0$ is given by
$$\mathcal R_0(u,v)=
\Biggl\langle
\mathbf{S}^\ast\circ\alpha \; \mbox{ \LARGE $\mid$  } \; \emph{\textrm{diag}}(\nu)\emph{\textrm{diag}}^{-1}((\mu+\nu)\circ(\mu+\gamma+\delta)) \mathbb M \emph{\diag}^{-1}(\mathbb M^T\mathbf{S}^\ast)  \mathcal B\circ p
\Biggr\rangle
$$

where $\alpha\in\R^u$ and $p\in\R^v$ are such that $\mathbb P=\alpha p^T$. Moreover, if $\mathbb P$ is stochastic,
\begin{eqnarray*}
\mathcal R_0(u,v)&=& \mathbf{S}^{\ast T}\emph{\textrm{diag}}(\nu)\emph{\textrm{diag}}^{-1}((\mu+\nu)\circ(\mu+\gamma+\delta))\mathbb M\emph{\textrm{diag}}(\mathcal B)\emph{\diag}^{-1}(\mathbb M^T\mathbf{S}^\ast)p
\\
&:=&
\Biggl\langle
\mathbf{S}^\ast \; \mbox{ \LARGE $\mid$  } \; \emph{\textrm{diag}}(\nu)\emph{\textrm{diag}}^{-1}((\mu+\nu)\circ(\mu+\gamma+\delta))\mathbb M \emph{\diag}^{-1}(\mathbb M^T\mathbf{S}^\ast)  \mathcal B\circ p
\Biggr\rangle
\end{eqnarray*}
\end{theorem}

\begin{proof}\hfill

If the susceptible residence times matrix  $\mathbb P$ is of rank one, there exists a vector $p\in\R^v$ and $\alpha\in\R^u$ such that $\mathbb P=\alpha p^T$. The next generation matrix is:
$$-FV^{-1}=\textrm{diag}(\mathbf{S^\ast})\mathbb M\textrm{diag}(\mathcal{B})\textrm{diag}^{-1}(\mathbb{M}^{T}\mathbf{S^\ast} )p\alpha^T\diag^{-1}((\mu+\nu)\circ(\mu+\gamma+\delta))$$
which is of rank one since it could be written as $xy^T$ where $x=\textrm{diag}(\mathbf{S^\ast})\mathbb M\textrm{diag}(\mathcal{B})\textrm{diag}^{-1}(\mathbb{M}^{T}\mathbf{S^\ast} )p$ and $y=\diag^{-1}((\mu+\nu)\circ(\mu+\gamma+\delta))\alpha$. Hence, its unique non zero eigenvalue is,
\begin{eqnarray}
\mathcal R_0(u,v)&=&\alpha^T\diag^{-1}((\mu+\nu)\circ(\mu+\gamma+\delta))\textrm{diag}(\mathbf{S^\ast})\mathbb M\textrm{diag}(\mathcal{B})\textrm{diag}^{-1}(\mathbb{M}^{T}\mathbf{S^\ast} )p\nonumber\\
&=&(\alpha\circ\mathbf{S^{\ast})^T}\diag^{-1}((\mu+\nu)\circ(\mu+\gamma+\delta))\mathbb M\textrm{diag}(\mathcal{B})\textrm{diag}^{-1}(\mathbb{M}^{T}\mathbf{S^\ast} )p\nonumber\\
&=&\Biggl\langle
\alpha\circ\mathbf{S}^\ast \; \mbox{ \LARGE $\mid$  } \; \diag(\nu)\diag^{-1}((\mu+\nu)\circ(\mu+\gamma+\delta)) \mathbb M \diag^{-1}(\mathbb M^T\mathbf{S}^\ast) \mathcal B\circ p
\Biggr\rangle\nonumber
\end{eqnarray}
If $\mathbb P$ is stochastic, we can show that $\alpha=\mathbf{1}$ and hence,
$$\mathcal R_0(u,v)=\Biggl\langle
\mathbf{S}^\ast \; \mbox{ \LARGE $\mid$  } \; \diag(\nu)\diag^{-1}((\mu+\nu)\circ(\mu+\gamma+\delta)) \mathbb M \diag^{-1}(\mathbb M^T\mathbf{S}^\ast) \mathcal B\circ p
\Biggr\rangle$$
which is the desired result.
\end{proof}
The condition of rank one of the matrices $\mathbb M$ and $\mathbb P$, when both matrices are stochastic, means that the susceptible and infected individuals of different groups spend the same proportion of time in each and every patch. When the matrices are not stochastic, the rank one condition means that the proportion of times spent by susceptible or infected individuals of different groups in each patch are proportional. That is, there exists $\alpha_j$ such that $m_{ij}=\alpha_j m_i$ for all $1\leq i,j\leq u$.

\section{Simulations}
In this section, we run some numerical simulations for 2 groups and 3 patches in order to highlight the effects of differential residence times and to illustrate the previously obtained theoretical results. To that end, unless otherwise stated, the baseline parameters of the model are chosen as follows:
$$
\beta_{1} = 0.25 \;\textrm{days}^{-1},\; \beta_{2} = 0.15 \;\textrm{days}^{-1},\; \beta_{3} = 0.1 \;\textrm{days}^{-1},\; \frac{1}{\mu_1}=75\times 365\;\textrm{days}, \;\frac{1}{\mu_2}=70\times 365 \;\textrm{days},$$
$$\Lambda_1 = 150,
\Lambda_2 = 100, \; \nu_1=\nu_2=\frac{1}{4} \;\textrm{days}^{-1},\; \frac{1}{\gamma_1}= 7\; \textrm{days},\; \frac{1}{\gamma_2}= 6\;  \textrm{days},
 \eta_1=\eta_2= 0.00137\; \textrm{days}^{-1},$$ $$ \delta_1=\delta_2= 2\times10^{-5}\; \textrm{days}^{-1}
$$
Although the values of $\beta_j$ are chosen throughout this section, for convenience, to be between 0 and 1, they need only to be nonnegative.
We begin by simulating the dynamics of Model \ref{PatchGenMatrix} when the basic reproduction number is below or above unity. Figure \ref{DynamicsI1I2} shows the dynamics of infected individuals of Group 1 (Fig. \ref{fig:I1Dynamics}) and Group 2 (Fig. \ref{fig:I2Dynamics}). The disease persists in both groups when $\mathcal R_0>1$ ( Fig. \ref{fig:I1Dynamics} and Fig. \ref{fig:I2Dynamics}, dotted red and dashed green curves) while it dies out when $\mathcal R_0<1$ ( Fig. \ref{fig:I1Dynamics} and Fig \ref{fig:I2Dynamics}, solid blue and dash-dotted black curves).\\

\begin{figure}
\centering     
\subfigure[Dynamics of $I_1$.]{\label{fig:I1Dynamics}\includegraphics[width=81mm]{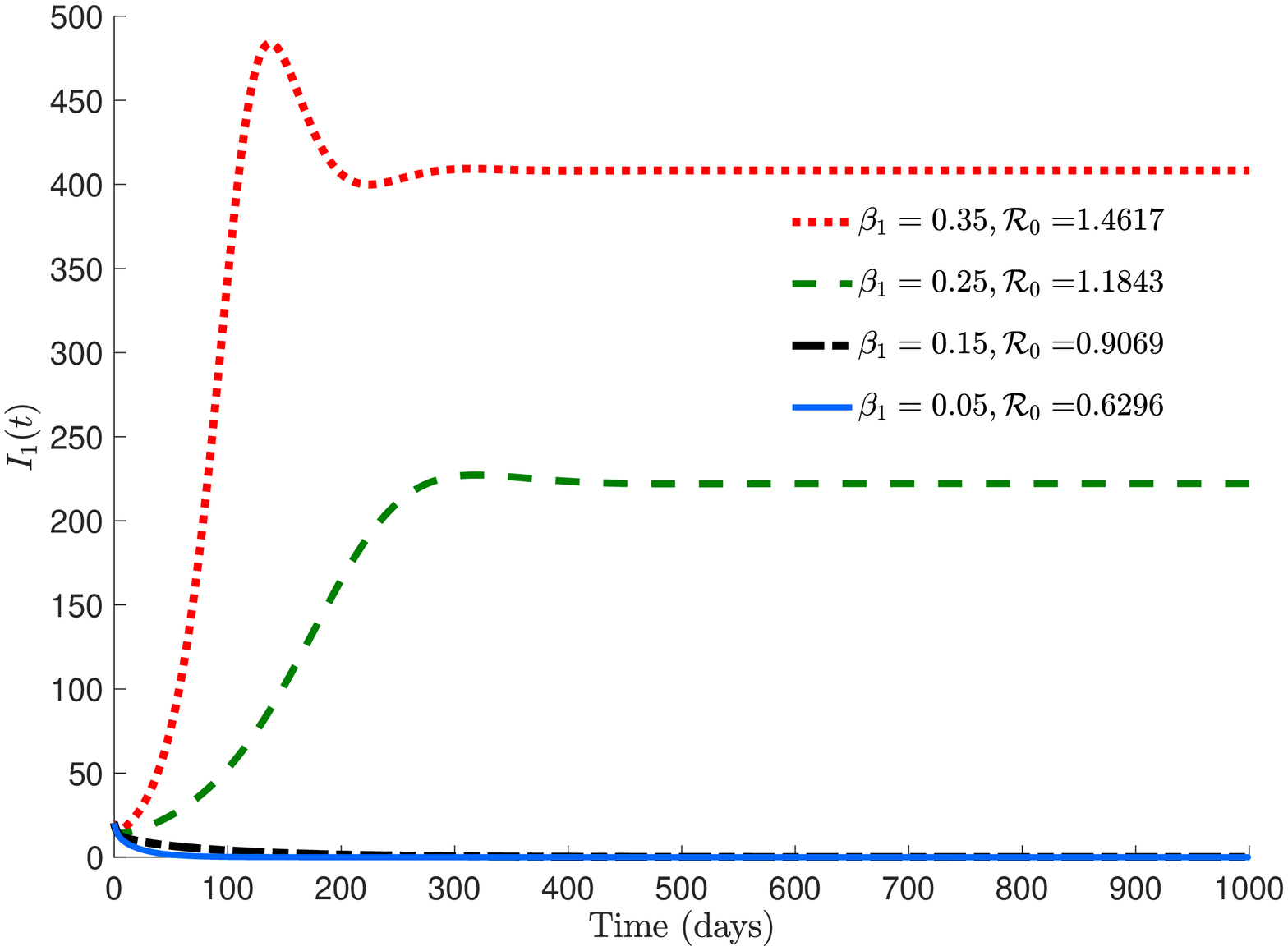}}
\subfigure[Dynamics of $I_{2}$.]{\label{fig:I2Dynamics}\includegraphics[width=81mm]{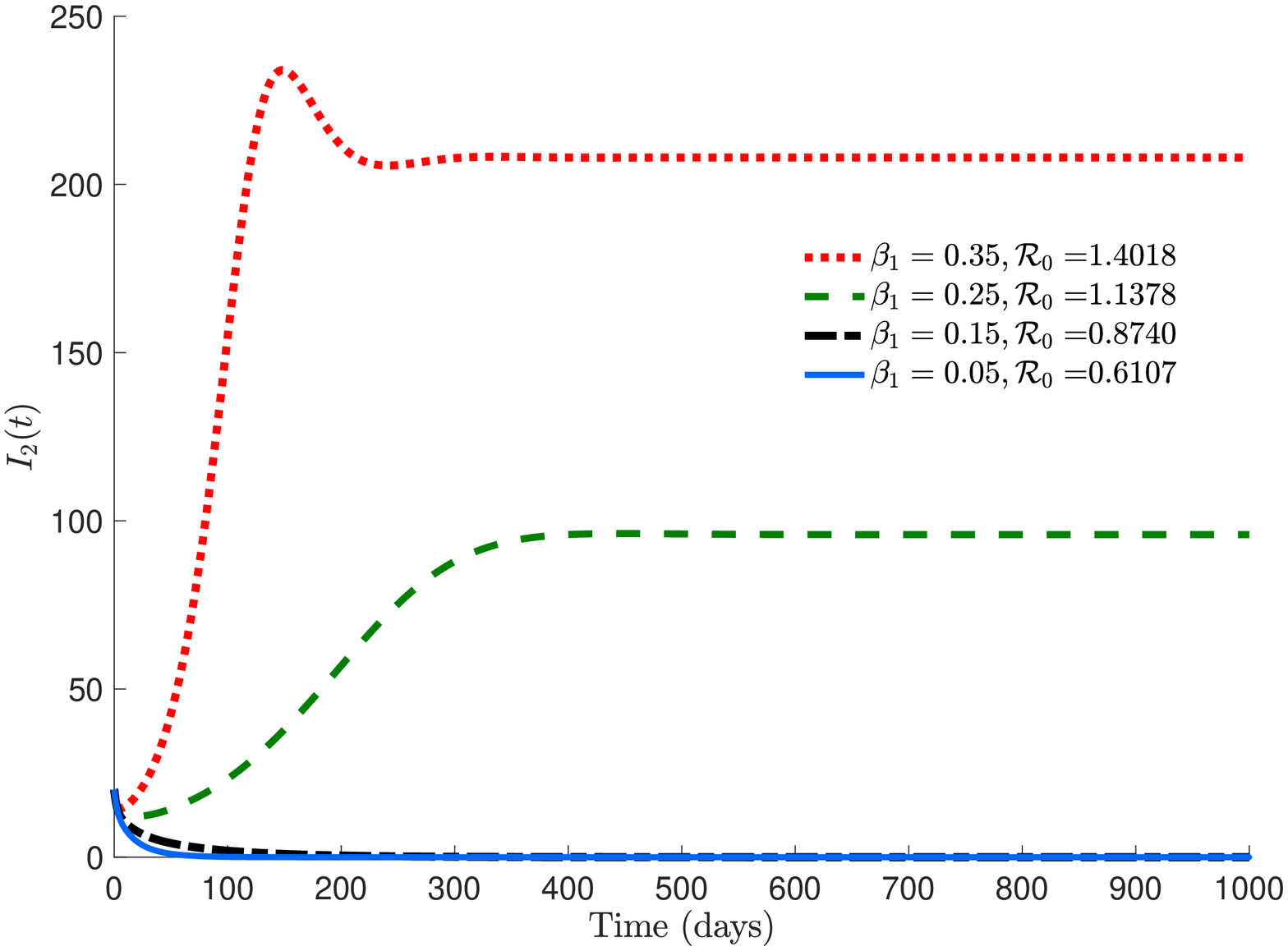}}
\caption{Dynamics of infected individuals of Group 1 (\ref{fig:I1Dynamics}) and Group 2 (\ref{fig:I2Dynamics}). Values of $\beta_{1}=0.35$, $\beta_2=0.25$, $\beta_3=0.15$ and $\mu_1=0.03$, and  $\mu_2=0.04$ are chosen for this set of simulations.}
\label{DynamicsI1I2}
\end{figure}
\begin{figure}
\centering     
\subfigure[$\bar I_1$ vs. $m_{11},\; p_{11}$.]{\label{fig:HeatI1m11vsp11Final}\includegraphics[width=81mm]{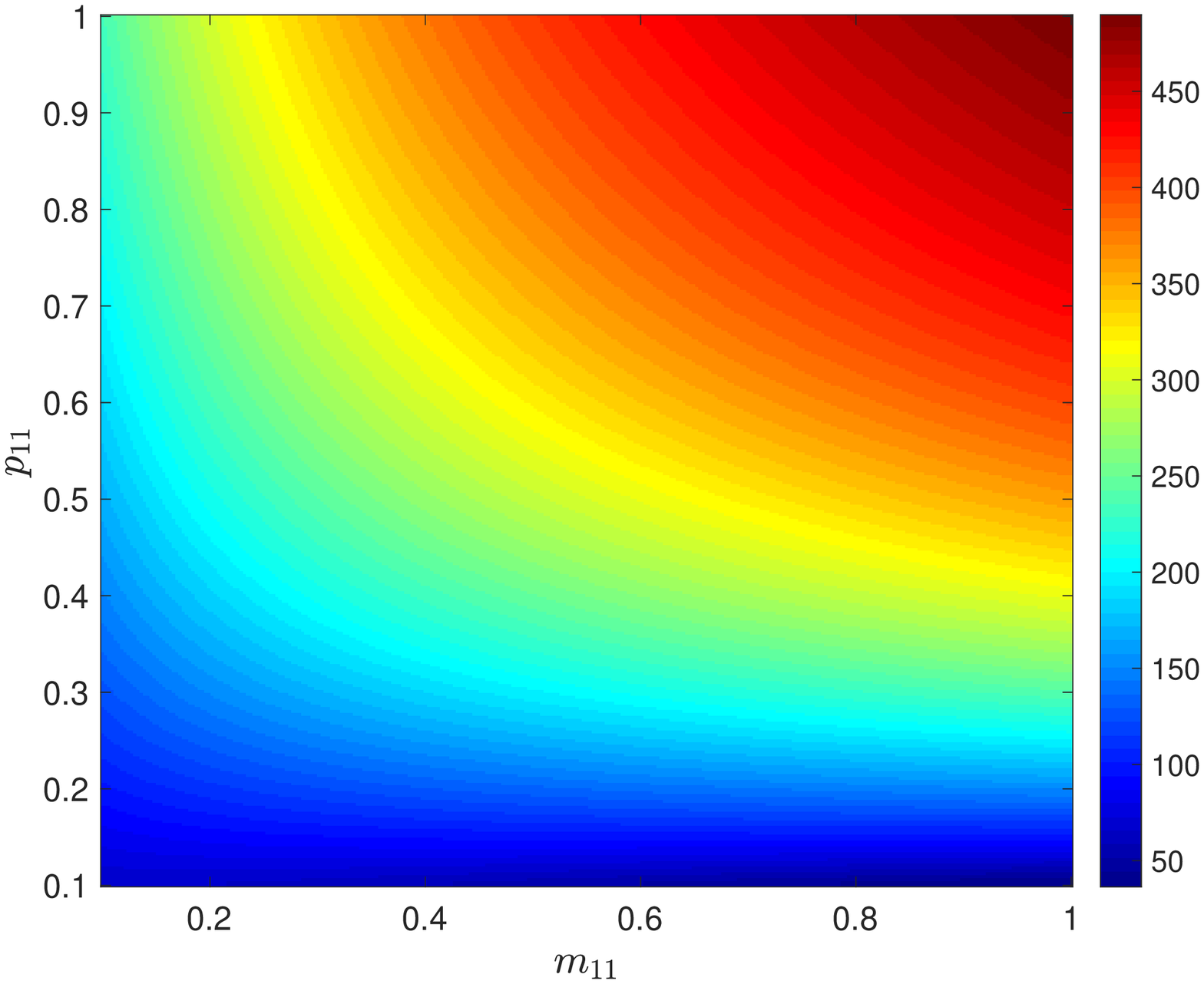}}
\subfigure[$\bar I_{2}$ vs. $m_{11},\; p_{11}$.]{\label{fig:HeatI2m11vsp11Final}\includegraphics[width=81mm]{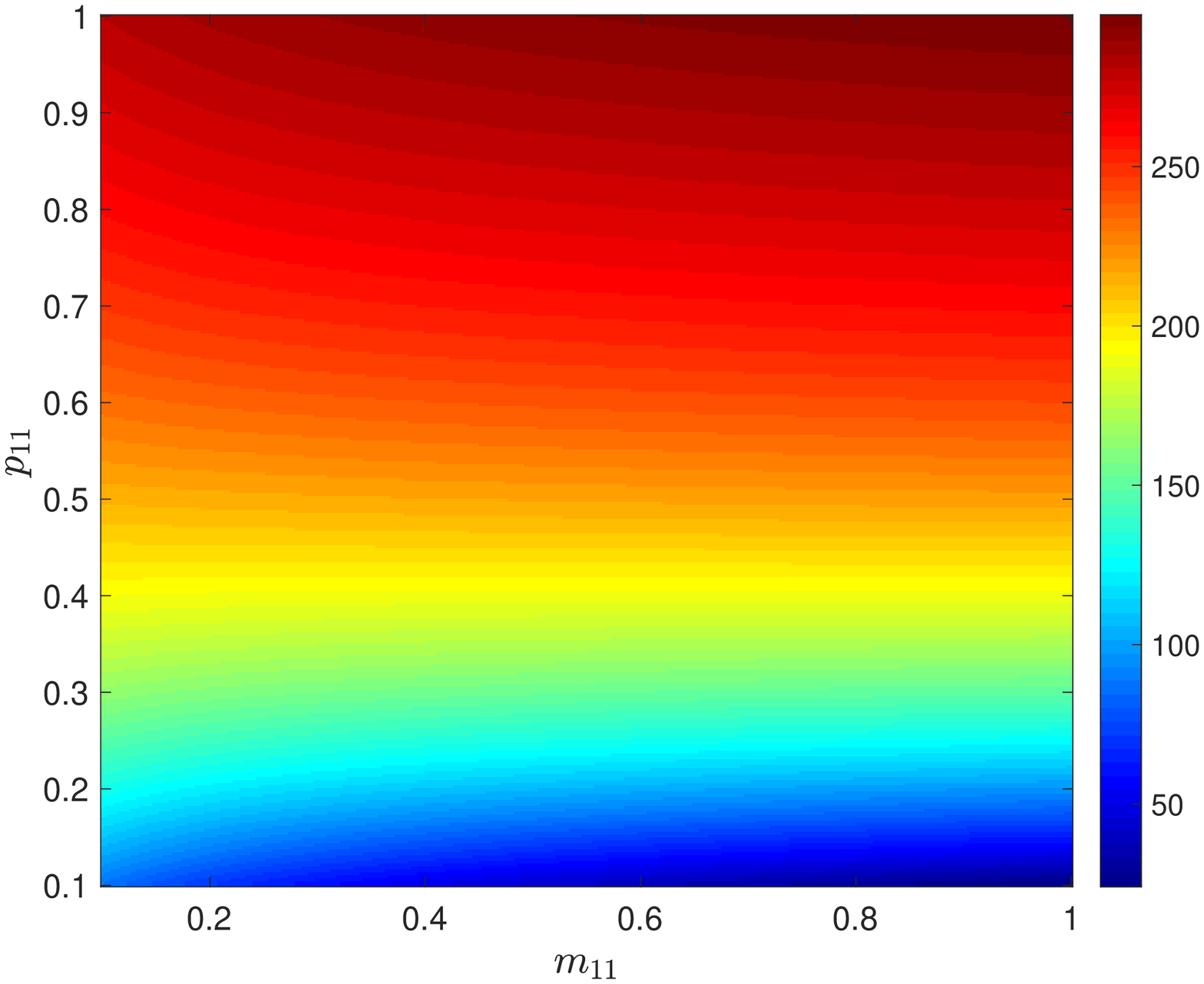}}
\caption{Variation of the disease prevalence at the equilibrium level with respect to the mobility patterns of susceptible and infected individuals of Group 1 (\ref{fig:HeatI1m11vsp11Final}) and Group 2 (\ref{fig:HeatI2m11vsp11Final}) in Patch 1 with $\beta_{1}=0.35$, $\beta_2=0.25$, $\beta_3=0.15$ and $\mu_1=\mu_2=0.05$.}
\label{I1I2vsM11P11}
\end{figure}
Figure (\ref{I1I2vsM11P11}) displays how the equilibrium value of infected individuals of Group 1 and Group 2 change with respect to residence times of infected and susceptible of group 1 in Patch 1, that is $m_{11}$ and $p_{11}$. For instance, in Fig. \ref{fig:HeatI1m11vsp11Final}, the disease burden in Group 1 ($\bar I_1$) is moderately low for all values of $m_{11}$ as long as $p_{11}$, the residence times of Group 1's infected into Patch 1, is below 0.3, even if Patch 1 is the riskiest patch with $\beta_1=0.35$. However, this prevalence level is more marked when $m_{11}\geq0.4$ and $p_{11}\geq0.5$.  The heatmap of $\bar{I_1}$ with respect to $m_{12}$ and $p_{21}$ shows similar patterns. We decided not to display this figure.  Fig \ref{fig:HeatI2m11vsp11Final} shows  the changes in the values of infected in  Group 2 ($\bar I_2$) due to movement patterns of susceptible and infected of Group 1 ($m_{11}$ and $p_{11}$) when their own movement patterns are fixed ($m_{21}=0.6$ and $p_{21}=0.4$).  \\

\begin{figure}
\centering     
\subfigure[$\mathcal R_0$ vs. $m_{11},\; p_{11}$.]{\label{fig:HeatR0m11vsp11}\includegraphics[width=81mm]{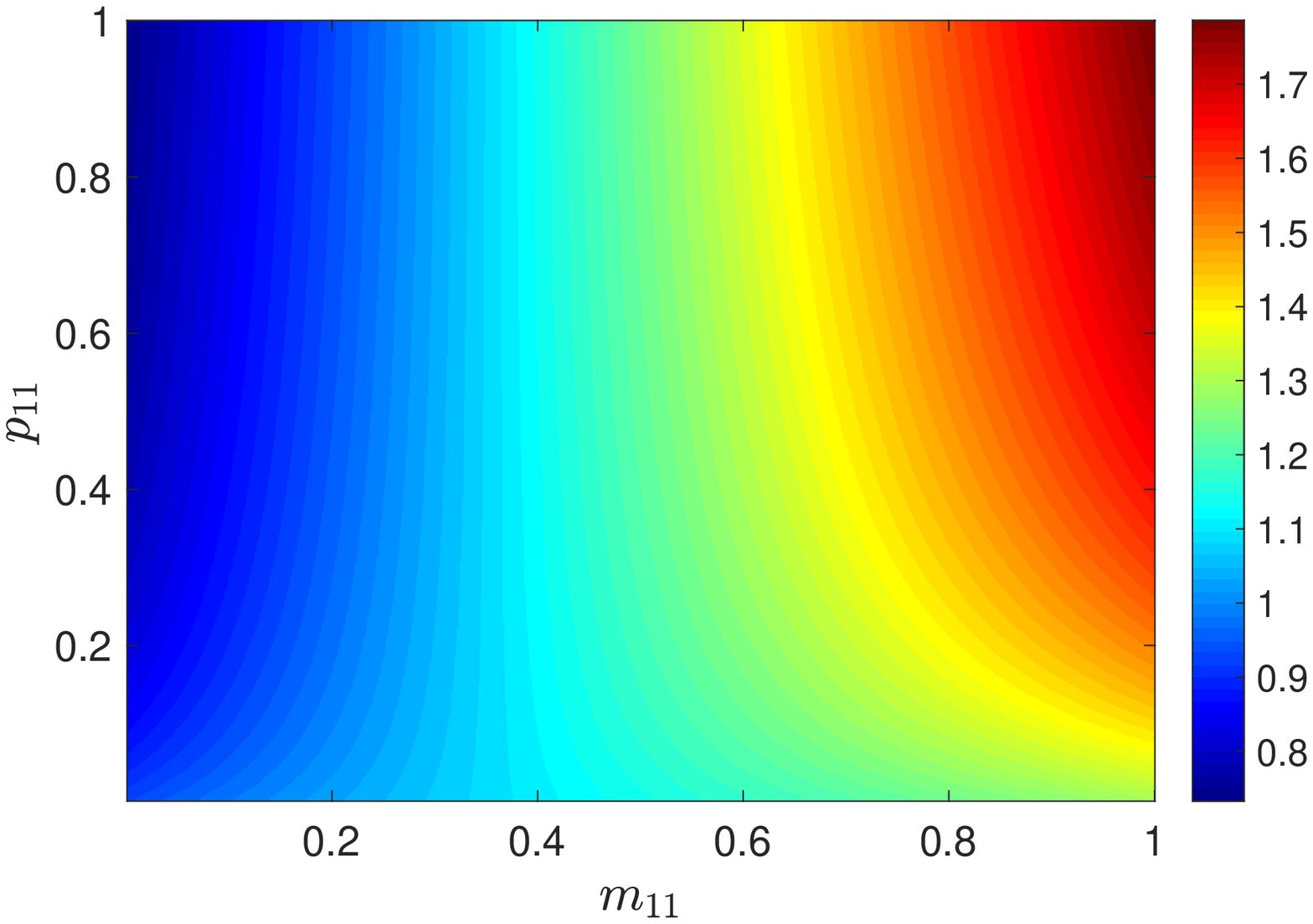}}
\subfigure[$\mathcal R_0$ vs. $m_{12},\; p_{12}$.]{\label{fig:HeatR0m12vsp12}\includegraphics[width=81mm]{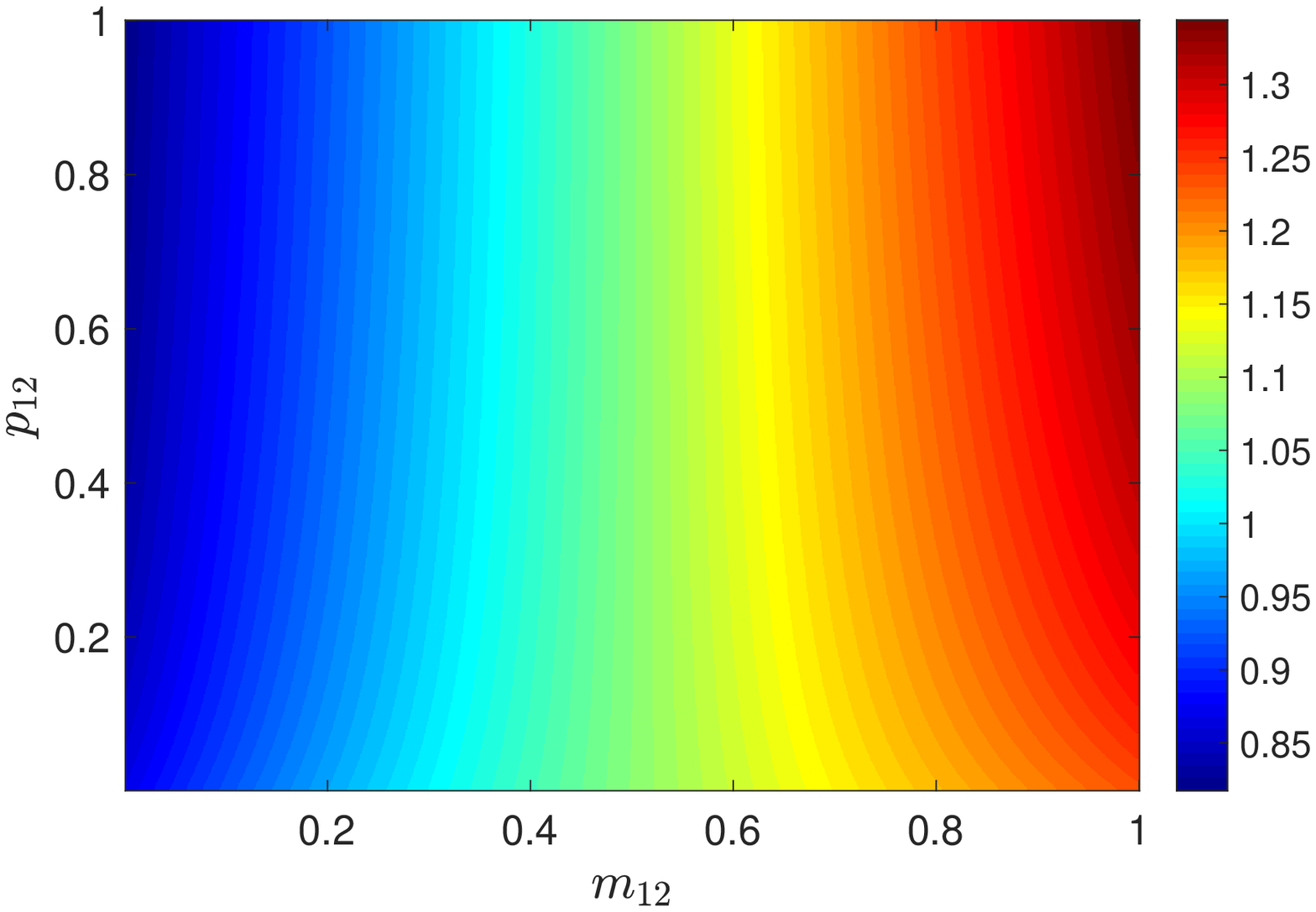}}
\caption{Variation of $\mathcal R_0$ with respect to the mobility patterns of susceptible and infected individuals of Group 1 in Patch 1 (\ref{fig:HeatR0m11vsp11}) and Patch 2 (\ref{fig:HeatR0m12vsp12}). Values of $\beta_{1}=0.2$, $\beta_2=0.1$ and $\beta_3=0.08$ are chosen for this set of simulations.}
\label{Fig:HeatR0vsM11P11}
\end{figure}

Fig  \ref{Fig:HeatR0vsM11P11} gives an overview of the dynamics of the basic reproduction number with respect of mobility patterns of susceptible and infected individuals of Group 1 in Patch 1 and Patch 2. Fig  \ref{fig:HeatR0m11vsp11} shows that $m_{11}$ and $p_{11}$ could  bring $\mathcal R_0$ from bellow unity to above unity. Particularly, if $m_{11}\geq0.4$, then $\mathcal R_0>1$, which lead to the persistence of the disease. Also, $\mathcal R_0$ is much higher when $m_{11}\geq0.7$ and $p_{11}\geq0.2$. Fig  \ref{fig:HeatR0m12vsp12} shows how $\mathcal R_0$ varies when the movement of infected and susceptible of Group 1 in Patch 2 change.\\

\begin{figure}
\centering     
\subfigure[$\mathcal R_0$ vs. $m_{11}$]{\label{fig:R0vsM11}\includegraphics[width=81mm]{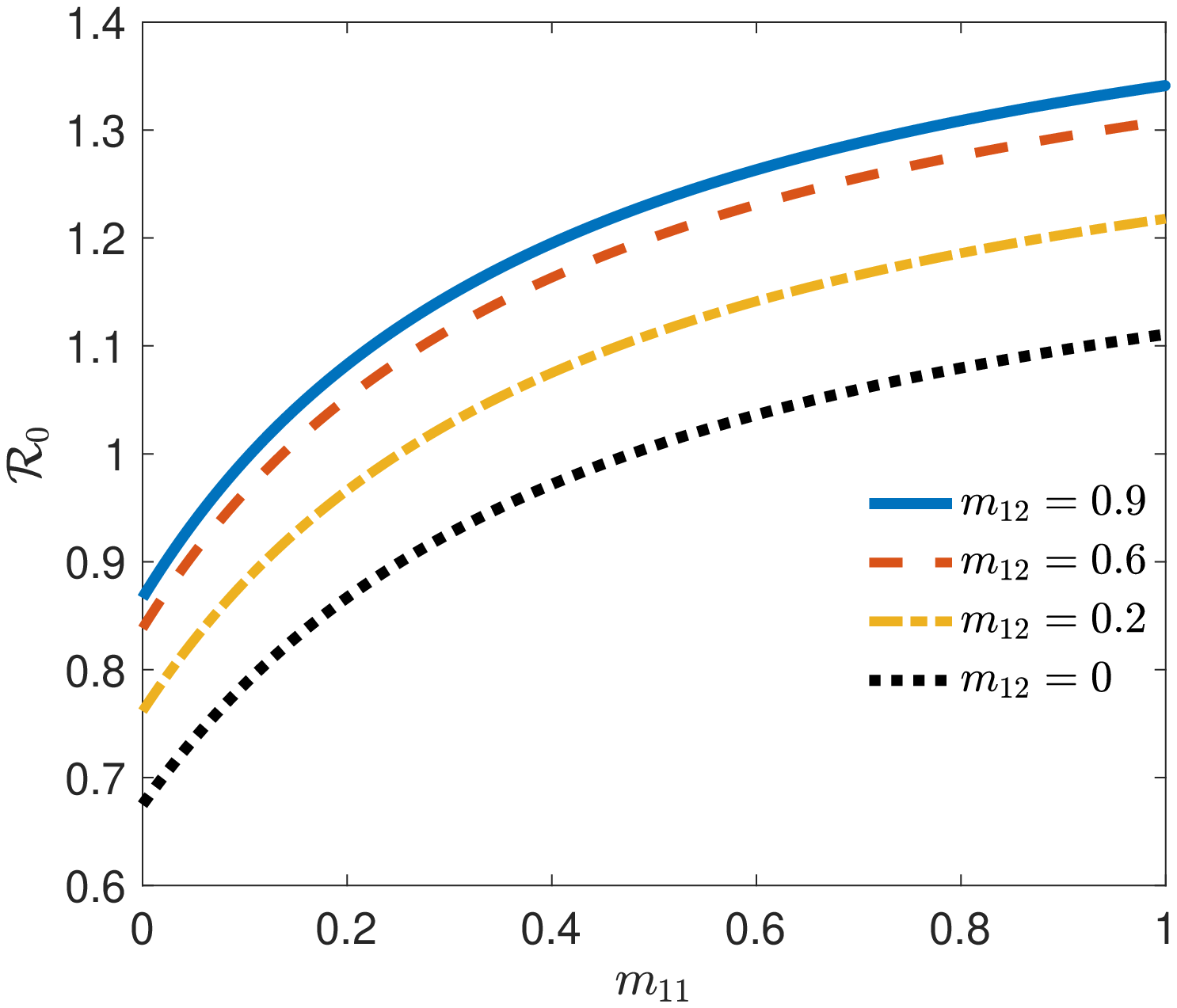}}
\subfigure[$\mathcal R_0$ vs. $p_{11}$]{\label{fig:R0vsP11}\includegraphics[width=81mm]{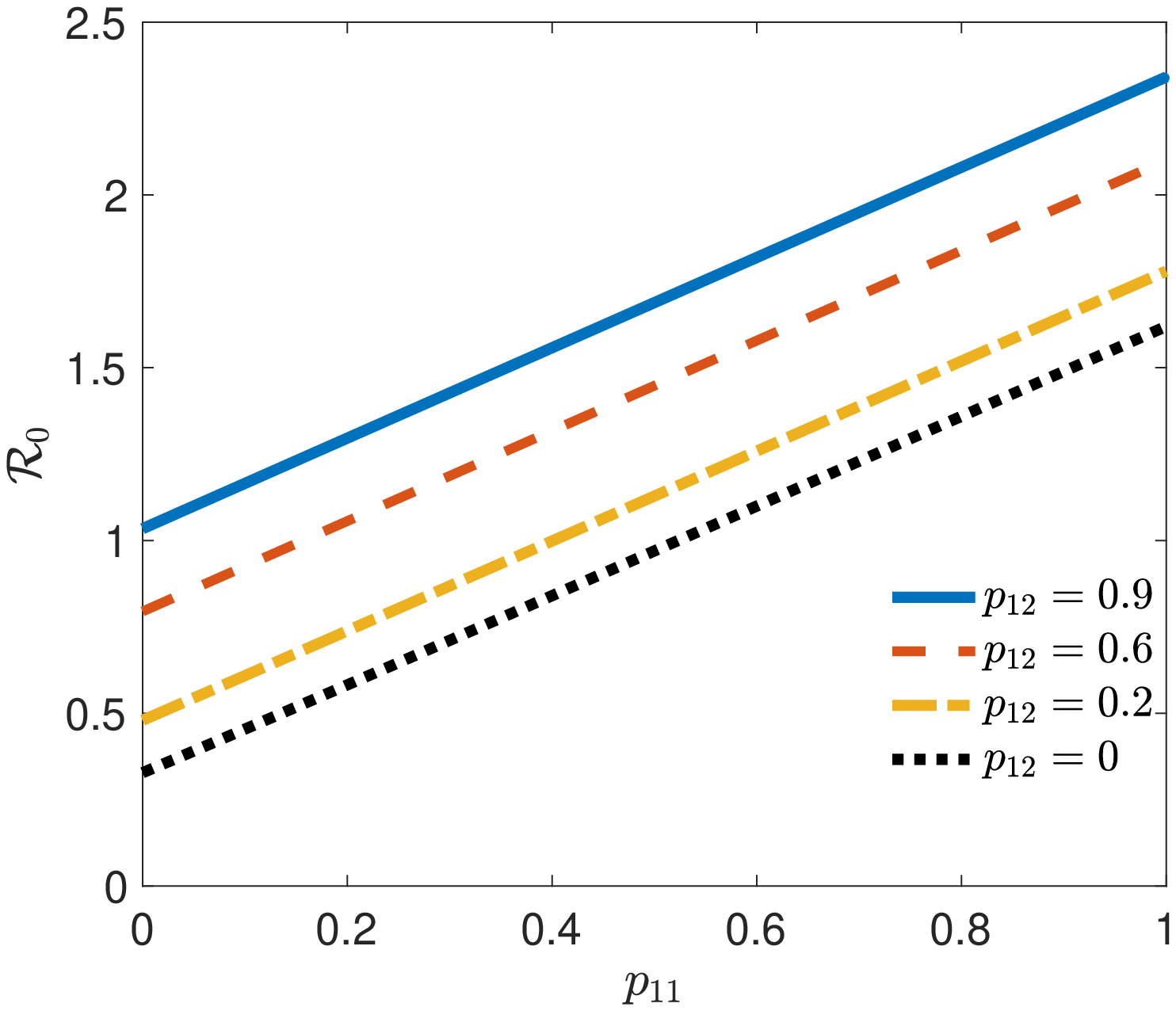}}
\caption{Variability of $\mathcal R_0$ with respect to $m_{11}$, $m_{12}$ and $p_{11}$, $p_{12}$. If all other parameters are fixed, $\mathcal R_0$ increases with respect to $m_{11}$ and $m_{12}$.}
\label{R0vsM11P11}
\end{figure}

In Figure (\ref{R0vsM11P11}), we revisit the variability of the basic reproduction number with respect of mobility patterns of susceptible and infected individuals of Group 1 (Fig \ref{Fig:HeatR0vsM11P11}). However, we obtain a clear  picture on how it changes.  Indeed, Fig. \ref{fig:R0vsM11}) suggests that $\mathcal R_0$ increases with respect to $m_{11}$ and $m_{12}$; and $p_{11}$ and $p_{12}$ (Fig. \ref{fig:R0vsP11}). However, $\mathcal R_0$ increases much faster with respect to $p_{11}$ than to $m_{11}$. Moreover, Fig \ref{fig:R0vsP11} confirms also the result of Theorem \ref{TheoremInegaR0P}, which states that the basic reproduction number increases with respect of $p_{ij}$, that is the movement patterns of infected individuals.

\begin{figure}[!tbp]
\centering
   \includegraphics[scale=0.4]{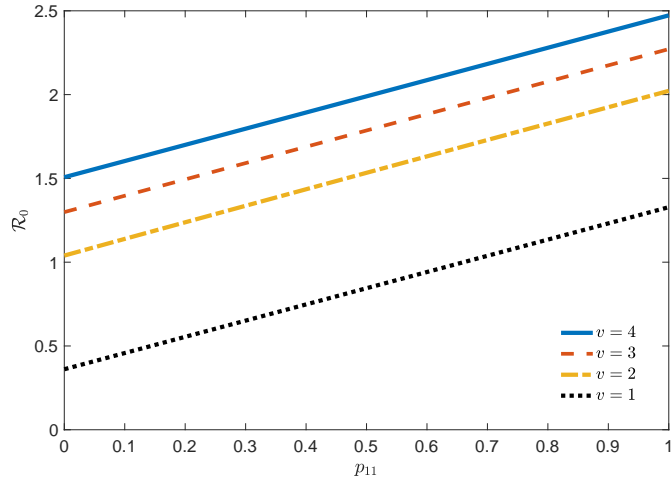}
   \caption{Effects of \textit{patchiness} on the basic reproduction number $\mathcal R_0$ with $u=3$. This risk of infection chosen for these 4 patches are: $\beta_1=0.25$, 
$\beta_2=0.15$,
$\beta_3=0.1$,
$\beta_4=0.08.$}
    \label{fig:R0IncreasesInPatches}
  \end{figure}

Fig (\ref{fig:R0IncreasesInPatches}) showcases that, for a fixed number of groups (3 in this case), the basic reproduction number increases as the number of patches increases, and that independently of the values of the risk of infection of the added patches. This figure, also confirms our the theoretical result in Item \ref{thm:Ineq3} of Theorem \ref{RTildeIneq}. It also shows a linear monotonicity of $\mathcal R_0(u,v)$ with respect to $\mathbb P$. Other values of $\beta$s than those of Fig. \ref{fig:R0IncreasesInPatches} exhibit similar patterns.
\section{Global stability of equilibria}\label{Sec:GAS}

The global stability of equilibria for the general Model (\ref{PatchGenMatrix}) happens to be very challenging. In fact, for models with such intricated nonlinearities, it is shown in  \cite{HuangCookeCCC92} that multiple endemic equilibria may exist. In this section, we explore the global stability of equilibria for some particular cases of the general model. 

\subsection{Identical Mobility and no disease induced  mortality}
In this subsection, we suppose that the host mobility to the patches is independent of the epidemiological status and that we neglect the disease induced mortality. In this case, the dynamics of the total population is given by 

$$\dot{\mathbf{N}}=\mathbf{\Lambda}-\mu\circ \mathbf{N}$$
Hence, $\displaystyle\lim_{t\to\infty}\mathbf{N}=\frac{\mathbf{\Lambda}}{\mu}:=\bar{\mathbf{N}}$. By using the theory of asymptotic systems \cite{CasThi95,0478.93044}, System (\ref{PatchGenMatrix}) is asymptotically equivalent to:
\begin{equation} \label{PatchGenMatrixSameResTimes}
\left\{\begin{array}{llll} 
\dot{\mathbf{S}}= {\mathbf{\Lambda}}-\textrm{diag}(\mathbf{S})\mathbb M\textrm{diag}(\mathcal{B})\textrm{diag}^{-1}(\mathbb{M}^{T}\bar{\mathbf{N}})\mathbb{M}^{T}\mathbf{I}  -\textrm{diag}(\mu) \mathbf{S}+\textrm{diag}(\eta)\mathbf{R}\\
\dot{\mathbf{E}}= \textrm{diag}(\mathbf{S})\mathbb M\textrm{diag}(\mathcal{B})\textrm{diag}^{-1}(\mathbb{M}^{T}\bar{\mathbf{N}})\mathbb{M}^{T}\mathbf{I}         -\textrm{diag}(\nu+\mu)\mathbf{E}\\
\dot{\mathbf{I}}=\textrm{diag}(\nu)\mathbf{E}- \textrm{diag}(\gamma+\mu)\mathbf{I} \\
\dot{\mathbf{R}}=\textrm{diag}(\gamma)\mathbf{I}- \textrm{diag}(\eta+\mu)\mathbf{R} 
\end{array}\right.
\end{equation}

Model (\ref{PatchGenMatrixSameResTimes}) generalizes models considered in \cite{bichara2015sis}. Let us denote $\mathcal R_0^{\textrm{Eq}}(u,v)$ the corresponding basic reproduction number of Model (\ref{PatchGenMatrixSameResTimes}). Its expression is $$\mathcal R_0^{\textrm{Eq}}(u,v)=\rho(\textrm{diag}(\mathbf{S^\ast})\mathbb M\textrm{diag}(\mathcal{B})\textrm{diag}^{-1}(\mathbb{M}^{T}\mathbf{S^\ast} )\mathbb{M}^{T} \textrm{diag}(\nu) \textrm{diag}^{-1}((\mu+\nu)\circ (\mu+\gamma)))$$
The following theorem gives the global stability of the disease free equilibrium.
\begin{theorem}\label{GlobalDFESameResTimes} \hfill

\n Whenever the host-patch mobility configuration $\mathbb M\mathbb M^T$ is irreducible, the following statements hold:
\begin{enumerate}
\item If $\mathcal R_0^{\textrm{Eq}}(u,v)\leq1$, the DFE is globally asymptotically stable (GAS). 
\item If $\mathcal R_0^{\textrm{Eq}}(u,v)>1$, the DFE is unstable.
\end{enumerate}
\end{theorem}
\begin{proof} \hfill

Let $(w_E,w_I)$ a left eigenvector of $Z \textrm{diag}(\nu) \textrm{diag}^{-1}((\mu+\nu)\circ (\mu+\gamma))$ corresponding to 
$\rho(Z \textrm{diag}(\nu) \textrm{diag}^{-1}((\mu+\nu)\circ (\mu+\gamma)))$  where $$Z=\textrm{diag}(\mathbf{S^\ast})\mathbb M\textrm{diag}(\mathcal{B})\textrm{diag}^{-1}(\mathbb{M}^{T}\mathbf{S^\ast} )\mathbb{M}^{T}$$

\n Hence,
\begin{eqnarray}
(w_E,w_I)Z \textrm{diag}(\nu) \textrm{diag}^{-1}((\mu+\nu)\circ (\mu+\gamma))&=&(w_E,w_I)\rho(Z \textrm{diag}(\nu) \textrm{diag}^{-1}((\mu+\nu)\circ (\mu+\gamma)))\nonumber\\
&=&(w_E,w_I)\rho(-FV^{-1})\nonumber
\end{eqnarray}
Since $\mathbb M\mathbb M^T$ is irreducible, the matrix $Z \textrm{diag}(\nu) \textrm{diag}^{-1}((\mu+\nu)\circ (\mu+\gamma))$ is irreducible. This implies that $(w_E,w_I)\gg0$.

We consider the Lyapunov function $$V(\mathbf{E},\mathbf{I})=(w_E,w_I) \left(\begin{array}{cccc}
\textrm{diag}^{-1}( \mu+\nu)& \textbf{0}_{u,u}\\ 
\textrm{diag}(\nu) \textrm{diag}^{-1}((\mu+\nu)\circ (\mu+\gamma)) & \textrm{diag}^{-1}(\mu+\gamma)
\end{array}\right)\left(\begin{array}{c}
 \mathbf{E}\\
 \mathbf{I}
\end{array}\right)$$

The derivative of $V(\mathbf{E},\mathbf{I})$ along trajectories of (\ref{PatchGenMatrixSameResTimes}) is
\begin{eqnarray}
\dot V(\mathbf{E},\mathbf{I})&=&(w_E,w_I) \left(\begin{array}{cccc}
\textrm{diag}( \mu+\nu)^{-1}& \textbf{0}_{u,u}\\ 
\textrm{diag}(\nu) \textrm{diag}^{-1}((\mu+\nu)\circ (\mu+\gamma)) & \textrm{diag}^{-1}(\mu+\gamma)
\end{array}\right)\left(\begin{array}{c}
 \dot{\mathbf{E}}\\
\dot{\mathbf{I}}
\end{array}\right)\nonumber\\
&=&(\tilde w_E, \tilde w_I) \left(\begin{array}{cccc}
-\textrm{diag}( \mu+\nu)& \textrm{diag}(\mathbf{S})\mathbb M\textrm{diag}(\mathcal{B})\textrm{diag}^{-1}(\mathbb{M}^{T}{\bar{\mathbf{N}}})\mathbb{M}^{T}\\ 
\textrm{diag}(\nu)  & -\textrm{diag}(\mu+\gamma) 
\end{array}\right)\left(\begin{array}{c}
 \mathbf{E}\\
\mathbf{I}
\end{array}\right)\nonumber
\end{eqnarray}
where $\tilde w_E=w_E\textrm{diag}^{-1}( \mu+\nu)+w_I\textrm{diag}(\nu) \textrm{diag}^{-1}((\mu+\nu)\circ (\mu+\gamma))$ and $\tilde w_I=w_I \textrm{diag}^{-1}(\mu+\gamma)$, or equivalently
$(\tilde w_E, \tilde w_I)=(w_E,w_I)(-V^{-1})$.\\

Since $ \textrm{diag}(\mathbf{S})\leq  \textrm{diag}(\mathbf{S^\ast})$ and $\mathbf{S^\ast}=\bar{\mathbf{N}}$, we obtain (denoting $\mathbb{I}$ the identity matrix),

\begin{eqnarray}
\dot V(\mathbf{E},\mathbf{I})&\leq&
(\tilde w_E, \tilde w_I )(F+V)\left(\begin{array}{c}
  \mathbf{E}\\
  \mathbf{I}
\end{array}\right)\nonumber\\
    &= &(w_E,w_I)\left(-V^{-1}F-\mathbb{I}\right)\left(\begin{array}{c}
  \mathbf{E}\\
  \mathbf{I}
\end{array}\right)\nonumber\\
 &= &\left(\mathcal R_0^{\textrm{Eq}}(u,v)-1\right)(w_E,w_I)\left(\begin{array}{c}
  \mathbf{E}\\
  \mathbf{I}
\end{array}\right)\nonumber\\
&\leq& 0.\nonumber
\end{eqnarray}
We consider first the case when $\mathcal R_0^{\textrm{Eq}}(u,v)<1$. Let $\mathcal E$ be an invariant set contained in $\Omega$, where $\dot V(\mathbf{E},\mathbf{I})=0$. This set is reduced to the origin of $\R^{2u}$ (i.e., $(\mathbf{E},\mathbf{I})=(0,0)$).  This, combined to the invariance of $\mathcal E$, leads to $\mathbf{R}=0$ and $\mathbf{S}=\mathbf{S}^\ast$. Hence, the only invariant set contained in  $\Omega$, such that $\dot V(\mathbf{E},\mathbf{I})=0$, is reduced to the DFE. Hence, by LaSalle's invariance principle \cite{Bhatia70,LaSLef61}, the DFE is globally asymptotically stable on $\Omega$. Since  $\Omega$ is an attracting set, we conclude that the DFE is GAS on the positive orthant $\R^{4u}_+$. \\
When $\mathcal R_0^{\textrm{Eq}}(u,v)=1$, we can show that 
\[\begin{array}{ll}
     \dot V(\mathbf{E},\mathbf{I}) = & (w_E+w_I\,\diag(\nu) \diag(\mu+\gamma+\delta)^{-1})\diag( \mu+\nu)^{-1}\,(\textrm{diag}(\mathbf{S})- \textrm{diag}(\mathbf{S^\ast}))\cdot   \\
      &   \mathbb M\textrm{diag}(\mathcal{B})\textrm{diag}^{-1}(\mathbb{M}^{T}{\bar{\mathbf{N}}})\mathbb{M}^{T}\, 
\mathbf{I} \\
& \leq 0 .
\end{array}
\]
Therefore, as above, LaSalle's invariance principle allows to conclude. \\

\n The instability of the DFE when $\mathcal R_0^{\textrm{Eq}}(u,v)>1$ follows from \cite{MR1057044,VddWat02}.
\end{proof}
The following theorem provides the uniqueness of the endemic equilibrium.\\

\begin{theorem}\label{UniquenessEE} \hfill

 If $\mathcal R_0^{\textrm{Eq}}(u,v)>1$,  Model (\ref{PatchGenMatrixSameResTimes}) has a unique endemic equilibrium. 
\end{theorem}

The proof of this theorem is similar to that of Theorem \ref{UniquenessEEBil} in the next subsection.

\subsection{{\textbf{\textit{Effective population size}}} dependent risk} \label{DDTransmission}

So far, the risk associated with each patch is represented by the constant vector $\mathcal B$. However, in some cases, it is more appropriate to assume that the risk of catching a disease depends on the size of the population or crowd, that is the \textit{effective population size} in each patch. In this subsection, we suppose that the risk of infection in each patch $j$ is linearly proportional to the \textit{effective population size}, that is $N_j^{\textrm{eff}}=\sum_{k=1}^u(m_{ij}S_i+n_{ij}E_i+p_{ij}I_i+q_{ij}R_i)$. Hence,

$$\beta_j(N_j^{\textrm{eff}})=\beta_j\sum_{k=1}^u(m_{kj}S_k+n_{kj}E_k+p_{kj}I_k+q_{kj}R_k)$$
Hence, the rate at which susceptible individuals  are infected in Patch $j$ is, therefore

$$\beta_j(N_j^{\textrm{eff}})\frac{\sum_{k=1}^up_{kj}I_k}{\sum_{k=1}^u(m_{kj}S_k+n_{kj}E_k+p_{kj}I_k+q_{kj}R_k)}:=\beta_j\sum_{k=1}^up_{kj}I_k$$
Therefore, in this settings, the dynamics of the population in different epidemiological classes take the form:

\begin{equation} \label{PatchGenBilin}
\left\{\begin{array}{llll}
\displaystyle\dot S_{i}= \Lambda_i-\sum_{j=1}^v\beta_jm_{ij}S_i\sum_{k=1}^up_{kj}I_k-\mu_iS_i+\eta_iR_i,\\
\displaystyle\dot E_i=\sum_{j=1}^v\beta_jm_{ij}S_i\sum_{k=1}^up_{kj}I_k-(\nu_i+\mu_i)E_i\\
\dot I_i=\nu_iE_i-(\gamma_i+\mu_i+\delta_i)I_i\\
\dot R_i=\gamma_iI_i-(\eta_i+\mu_i)R_i
\end{array}\right.
\end{equation}
System (\ref{PatchGenBilin}) could be written in a compact form as follows:
\begin{equation} \label{PatchGenMatrixBil}
\left\{\begin{array}{llll}
\dot{\mathbf{S}}= {\mathbf{\Lambda}}-\textrm{diag}(\mathbf{S})\mathbb M\textrm{diag}(\mathcal{B})\mathbb{P}^{T}\mathbf{I}  -\textrm{diag}(\mu) \mathbf{S}+\textrm{diag}(\eta)\mathbf{R}\\
\dot{\mathbf{E}}=\textrm{diag}(\mathbf{S})\mathbb M\textrm{diag}(\mathcal{B})\mathbb{P}^{T}\mathbf{I} -\textrm{diag}(\nu+\mu)\mathbf{E}\\
\dot{\mathbf{I}}=\textrm{diag}(\nu)\mathbf{E}- \textrm{diag}(\gamma+\mu+\delta)\mathbf{I} \\
\dot{\mathbf{R}}=\textrm{diag}(\gamma)\mathbf{I}- \textrm{diag}(\eta+\mu)\mathbf{R} 
\end{array}\right.
\end{equation}

Clearly, System (\ref{PatchGenMatrixBil}) is a particular case of System (\ref{PatchGenMatrix}) when the transmission term takes a modified density-dependent form. Positivity and boundedness properties of solutions of System (\ref{PatchGenMatrix}) hold for those of System (\ref{PatchGenMatrixBil}). The basic reproduction number of Model (\ref{PatchGenMatrixBil}), denoted by $\mathcal R_0^{\textrm{DD}}(u,v) $ is:

$$\mathcal R_0^{\textrm{DD}}(u,v) =\rho(\textrm{diag}(\mathbf{S^\ast})\mathbb M\textrm{diag}(\mathcal{B})\mathbb{P}^{T}\textrm{diag}(\nu) \textrm{diag}^{-1}((\mu+\nu)\circ (\mu+\gamma+\delta)))$$

We explore the properties of steady state solutions. The following result gives the global stability of the DFE. Its proof is similar to that of Theorem \ref{GlobalDFESameResTimes}.

\begin{corollary} \hfill

\n Whenever the host-patch mobility configuration $\mathbb M\mathbb P^T$ is irreducible, the following statements hold:
\begin{enumerate}
\item If $\mathcal R_0^{\textrm{DD}}(u,v)\leq1$, the DFE is globally asymptotically stable. 
\item If $\mathcal R_0^{\textrm{DD}}(u,v)>1$, the DFE is unstable.
\end{enumerate}
\end{corollary} 

The proof of the existence and uniqueness of the endemic equilibrium (EE) for Model (\ref{PatchGenMatrixBil}) is done in two steps, by carefully crafting a new auxiliary system whose EE uniqueness is tied to that of Model (\ref{PatchGenMatrixBil}). 

Let 
\begin{multline}\label{ALK} A=\diag^{-1}(\eta+\mu)\,\diag(\gamma)\,\diag^{-1}(\gamma+\mu+\delta)\,\diag(\nu),\quad L=\diag^{-1}(\gamma+\mu+\delta)\,\diag(\nu)\\\textrm{and}\quad \quad K=\diag^{-1}(\mu)\diag(\nu+\mu)-\diag^{-1}(\mu)\,\diag(\eta)\,A\quad\quad\quad\end{multline}      

We have the following lemma,
\begin{lemma}\label{GEE}\hfill 

Model (\ref{PatchGenMatrixBil}) has a unique endemic equilibrium if the function
$$g(y)=\diag^{-1}(\nu+\mu)\,
\diag(S^*-Ky)\mathbb M\diag(\mathcal{B})
\mathbb{P}^{T}\,Ly,
$$
  has  a unique fixed point.
\end{lemma}
\begin{proof}\hfill

Let ($\mathbf{\bar S}, \mathbf{\bar E}, \mathbf{\bar I}, \mathbf{\bar R}$) an equilibrium point of System (\ref{PatchGenMatrixBil}) with $\mathbf{\bar I}\gg 0$. This equilibrium satisfies the following system:

\begin{equation} \label{EquilibriumRelations}
\left\{\begin{array}{llll}
\mathbf{0}= {\mathbf{\Lambda}}-\textrm{diag}(\mathbf{\bar S})\mathbb M\textrm{diag}(\mathcal{B})\mathbb{P}^{T}\mathbf{\bar I}  -\textrm{diag}(\mu) \mathbf{\bar S}+\textrm{diag}(\eta)\mathbf{\bar R}\\
\mathbf{0}= \textrm{diag}(\mathbf{\bar S})\mathbb M\textrm{diag}(\mathcal{B})\mathbb{P}^{T}\mathbf{\bar I}         -\textrm{diag}(\nu+\mu)\mathbf{\bar E}\\
\mathbf{0}=\textrm{diag}(\nu)\mathbf{\bar E}- \textrm{diag}(\gamma+\mu+\delta)\mathbf{\bar I} \\
\mathbf{0}=\textrm{diag}(\gamma)\mathbf{\bar I}- \textrm{diag}(\eta+\mu)\mathbf{\bar R} 
\end{array}\right.
\end{equation}
We can easily see that $\mathbf{\bar R} =A\mathbf{\bar E}$ and $\mathbf{\bar I}=L\mathbf{\bar E}$, where $A$, $L$ and $K$ are as defined in (\ref{ALK}). Thus, $\mathbf{\bar I}\gg0$ implies that $\mathbf{\bar E}\gg0$ and $\mathbf{\bar R}\gg0$. \\

\n Hence, System (\ref{EquilibriumRelations}) could be written only in terms of $\mathbf{\bar S}$ and $\mathbf{\bar E}$, that is:
\begin{equation}\label{SE}
\left\{
\begin{array}{l}
\mathbf{\bar S}=\diag^{-1}(\mu)\,\Big(\Lambda-\diag(\mathbf{\bar S})\mathbb M\diag(\mathcal{B})\mathbb{P}^{T}\,L\mathbf{\bar E}+\diag(\eta)\,A\mathbf{\bar E}\Big)\\
\mathbf{\bar E}=\diag^{-1}(\nu+\mu)\,\diag(\mathbf{\bar S})\mathbb M\diag(\mathcal{B})\mathbb{P}^{T}\,L\,\mathbf{\bar E}
    \end{array}
\right.
\end{equation}

Let $x=\diag^{-1}(\mu)\Lambda-\mathbf{\bar S}$ and $y=\mathbf{\bar E}$. Since $\mathbf{\bar S}\in\Omega$, it is clear that $x\geq\mathbf{0}$ and $y\geq\mathbf{0}$. Expressing the system (\ref{SE}) into new variables, we obtain:

    \begin{subequations}\label{eq:XY}

  \begin{align}[left = \empheqlbrace\,]  
    & x=\diag^{-1}(\mu)\,f(x,y) -\diag^{-1}(\mu)\,\diag(\eta)\,A\,y \label{eq:XYa}\\
     &y=\diag^{-1}(\nu+\mu)\,f(x,y)    \label{eq:XYb}
 \end{align}
    \end{subequations}

where \[f(x,y)=\diag(S^*-x)\mathbb M\diag(\mathcal{B})\mathbb{P}^{T}\,Ly
\]
It follows from (\ref{eq:XYb}) that $f(x,y)=\diag(\nu+\mu)\,y$, and hence (\ref{eq:XYa}) implies that $x=Ky$ where $$K=\diag^{-1}(\mu)\diag(\nu+\mu)-\diag^{-1}(\mu)\,\diag(\eta)\,A$$

After some algebraic manipulations, it could be shown that $K>0$. Combining the fact that $x=Ky$ and 
(\ref{eq:XYb}), it follows that  (\ref{eq:XY}), and subsequently (\ref{EquilibriumRelations}), could be written in the single vectorial equation:

$$y=g(y)$$ where 
\begin{eqnarray}
g(y)&=& \diag^{-1}(\nu+\mu)\,f(Ky,y)\nonumber\\
&=&\diag^{-1}(\nu+\mu)\,\diag(S^*-Ky)\mathbb M\diag(\mathcal{B})\mathbb{P}^{T}\,Ly
\end{eqnarray}
Thus, Model (\ref{PatchGenMatrixBil}) has a unique endemic equilibrium $\mathbf{\bar I}\gg0$ if and only if $g(y)$ has a unique fixed point $\bar y\gg0$. The desired result is achieved.
\end{proof}

Next, we present another lemma whose proof is straightforward:
\begin{lemma}\label{GFrelationship}
The function $g(y)$ has a fixed point $\bar y$ if and only if $\bar y$ is an equilibrium of $\dot{y}=F(y)$ where
$$F(y)=\diag(\nu+\mu)g(y)-\diag(\nu+\mu)y$$
\end{lemma}
The proof of this lemma is straightforward.
\begin{theorem}\label{UniquenessEEBil}\hfill

Under the assumption that the host-patch mobility configuration $\mathbb M\mathbb P^T$ is irreducible, Model  (\ref{PatchGenMatrixBil}) has a unique endemic equilibrium whenever $\mathcal R_0^{\textrm{DD}}(u,v)>1$.
\end{theorem}
\begin{proof}\hfill

\n By using Lemma \ref{GEE} and Lemma \ref{GFrelationship}, the uniqueness of EE for Model  (\ref{PatchGenMatrixBil}) is equivalent to the uniqueness of an EE for this system 
\begin{equation}\label{AuxSys}\dot y=F(y)\end{equation} when $\mathcal R_0^{\textrm{DD}}(u,v)>1$. Therefore, we will prove that the auxiliary system (\ref{AuxSys}) has an unique EE. In fact, we will prove that this equilibrium is globally attractive if $\mathcal R_0^{\textrm{DD}}(u,v)>1$. The proof of the latter is based on Hirsch's theorem \cite{Hirsch84}, by using elements of monotone systems. The Jacobian of the vector field $F(y)$ is:
\[\begin{array}{l}
F'(y)=\diag(\nu+\mu)\,(g'(y)- \mathbb{I}) \\
=
\left(-\diag\Big(\mathbb M\diag(\mathcal{B})\,\mathbb{P}^{T}\,Ly\Big)\,K
+\diag(S^*-Ky)\mathbb M\diag(\mathcal{B})\,\mathbb{P}^{T}\,L\right)
-\diag(\nu+\mu)\,\mathbb{I}
 \\
=-\diag(\nu+\mu)\,    \mathbb{I}-
\diag\Big(\mathbb M\diag(\mathcal{B})\,\mathbb{P}^{T}\,Ly\Big)\,K
+\diag(S^*-Ky)\mathbb M\diag(\mathcal{B})\,\mathbb{P}^{T}\,L.
\end{array}\]
where $\mathbb{I}$ is the identity matrix. The matrix $-\diag(\nu+\mu)\, \mathbb{I}-
\diag\Big(\mathbb M\diag(\mathcal{B})\,\mathbb{P}^{T}\,Ly\Big)\,K$ is a diagonal matrix and $\diag(S^*-Ky)\mathbb M\diag(\mathcal{B})\,\mathbb{P}^{T}\,L$ is a nonnegative matrix (since $S^*-Ky= \mathbf{\bar S}$). It follows that $F'(y)$ is Metzler and is irreducible since
$\mathbb M\,\mathbb{P}^{T}$ is. Therefore, System (\ref{AuxSys}) is strongly monotone. Moreover, it is clear that the map 
$F' : \R^{u}\longrightarrow \R^{u}\times\R^{u}$ is antimonotone.
Also, $F(0_{\R^u})=0_{\R^u}$ and $F'(0_{\R^u})=g'(0_{\R^u})-\mathbb{I}=\diag(S^*)\mathbb M\diag(\mathcal{B})\,\mathbb{P}^{T}\,L -\mathbb{I}$. Since $\rho(g'(0_{\R^u}))=\rho(\diag(S^*)\mathbb M\diag(\mathcal{B})\,\mathbb{P}^{T}\,L)=\mathcal R_0^{\textrm{DD}}(u,v) >1$, we deduce that $F'(0_{\R^u})$ has at least a positive eigenvalue
and therefore $0_{\R^u}$ is unstable. Therefore,   
System (\ref{AuxSys}) has unique equilibrium $\bar{y}\gg 0_{\R^u}$, which is globally attractive, due to Hirsch's theorem \cite{Hirsch84} (Theorem 6.1). We conclude that Model (\ref{PatchGenMatrixBil}) has a unique endemic equilibrium whenever $\mathcal R_0^{\textrm{DD}}(u,v) >1$. 
\end{proof}
Note that with the choice of $\mathbb P=\mathbb M\textrm{diag}^{-1}(\mathbb M^{T}\bar{\mathbf N})$ and $\delta=0$, System (\ref{PatchGenMatrixBil}) is exactly System (\ref{PatchGenMatrixSameResTimes}). 
Therefore, their solutions have the same asymptotic behavior.
\section{Conclusion and discussions}\label{sec:ConclusionDiscussions}

Heterogeneity in space and social groups are often studied separately and sometimes interchangeably in the context of disease dynamics.  Moreover, in these settings, susceptibility of the infection is based on group or individual. In this paper, we propose a new framework that incorporates heterogeneity in space and in group for which the structure of the latter is independent from that of the former. We define patch as a location where the infection takes place, which has a particular risk of infection. This risk is tied to environmental or hygienic or economic conditions that favors the infection. The likelihood of infection in each patch depends on both the risk of the patch and the proportion of time each host spend in that environment. We argue that this patch-specific risk is easier to assess compared to the classical differential susceptibility or WAIFW matrices.  Human host is structured in groups, where a \textit{group} is defined as a collection of individuals with similar demographic, genetic or social characteristics. In this framework, the population of each patch at time $t$ is captured by the temporal mobility patterns of all host groups visiting the patches, which in turn depends on the host's epidemiological status.   \\

Under this framework, we propose a general SEIRS multi-patch and multi-group model with differential state-host mobility patterns. We compute the basic reproduction number of the system with $u$ groups and $v$ patches, $\mathcal R_0(u,v)$, which depends on the mobility matrices of susceptible, $\mathbb M$, and infected, $\mathbb P$. The disease persists when $\mathcal R_0>1$ and dies out from all patches when $\mathcal R_0(u,v)<1$ (Fig. \ref{DynamicsI1I2}), when $\mathbb M\mathbb P^T$ is irreducible. When this matrix is not irreducible, the disease will persist or die out in all patches of the subsystem for which the configuration group-patch is irreducible and will be decoupled from the remaining system. \\

We systematically investigate the effects of heterogeneity in mobility patterns, groups and patches on the basic reproduction and on disease prevalence. Indeed, we have shown that, if the epidemiological parameters are fixed, the basic reproduction number is an increasing function of the entries of infected  hosts' movement matrix (e.g. Theorem \ref{TheoremInegaR0P}). Also, if the number of groups is fixed, an increase in the number of patches  increases the basic reproduction number (e.g. see Theorem \ref{R0Bounds}).  Explicit expressions of the basic reproduction numbers are obtained when the mobility matrices $\mathbb M$ and $\mathbb P$ are of rank one. That is, when, for all groups, susceptible (and infected) individuals' residence times in all patches are proportional (Theorems \ref{Mrank1} and \ref{Prank1}). It turns out that if the susceptible residence time matrix is of rank one and stochastic, the basic reproduction number is independent of $\mathbb M$. Moreover, we also show that if $\mathbb M$ is of rank one, its stochasticity is sufficient but not necessary for the basic reproduction number to be independent of $\mathbb M$. However, if the infected residence time matrix $\mathbb P$ is of rank one, stochastic or otherwise, the basic reproduction number still depends on the infected movement patterns.\\

The patch-specific risk vector $\mathcal B$ could also depend on the effective population size. We explored the case when this dependence is linear, that is when, for each patch $j$, $\beta_j(N_j^{\textrm{eff}})=\beta_j\sum_{k=1}^u(m_{kj}S_k+n_{kj}E_k+p_{kj}I_k+q_{kj}R_k)$. In this case, the transmission term of our model is captured by a density dependent incidence. Moreover, we show that this case is isomorphic to the general model, where the mobility patterns of host does not dependent on the epidemiological class, that is when $\mathbb M=\mathbb N=\mathbb P=\mathbb Q$. We prove that, in this case the disease free equilibrium is globally asymptotically stable whenever $\mathcal R_0\leq1$ while an unique endemic equilibrium exists if $\mathcal R_0>1$.  \\

\n We suspect that the disease free equilibrium is globally asymptotically stable whenever $\mathcal R_0\leq1$ for Model (\ref{PatchGenMatrix}), where the patch-specific risk is constant. A similar remark holds for the global stability of the endemic equilibrium of 
Model~(\ref{PatchGenMatrixBil}) and Model~(\ref{PatchGenMatrix}) when 
$\mathcal R_0>1$. This is still under investigation. Further areas of extensions of this study include more general forms of the patch-specific risks and when mobility patterns reflect the choices that individuals make at each point in time. These choices are based on maximizing the discounted value of an economic criterion \`a la \cite{EliCCC2011,Perrings:2014yq}.

      \subsection*{Acknowledgements}
{\small{We are grateful to two anonymous referees and the handling editor Dr. Gabriela Gomes for valuable comments and suggestions that led to an
improvement of this paper. We also thank  Bridget K. Druken for the careful reading and constructive comments. 
A. Iggidr acknowledges the partial support of Inria  in the framework of the program STIC AmSud (project MOSTICAW).
 }}

\begin{thebibliography}{10}

\bibitem{AndMay91}
{\sc R.~M. Anderson and R.~M. May}, {\em {Infectious Diseases of Humans.
  Dynamics and Control}}, Oxford science publications, 1991.

\bibitem{Arino08}
{\sc J.~Arino}, {\em Disease in metapopulations model}, in Modeling and
  dynamics of infectious diseases, Z.~Ma, Y.~Zhou, and J.~Wu, eds., World
  Scientific Publishing, 65-123~ed., 2009.

\bibitem{ArinoPortet2015}
{\sc J.~Arino and S.~Portet}, {\em Epidemiological implications of mobility
  between a large urban centre and smaller satellite cities}, Journal of
  Mathematical Biology, 71 (2015), pp.~1243--1265.

\bibitem{ArVdd06}
{\sc J.~Arino and P.~van~den Driessche}, {\em Disease spread in
  metapopulations}, in Nonlinear dynamics and evolution equations, X.-O. Zhao
  and X.~Zou, eds., vol.~48, Fields Instit. Commun., AMS, Providence, R.I.,
  2006, pp.~1--13.

\bibitem{0815.15016}
{\sc A.~Berman and R.~J. Plemmons}, {\em {Nonnegative matrices in the
  mathematical sciences.}}, SIAM, 1994.

\bibitem{Bhatia70}
{\sc N.~P. Bhatia and G.~P. Szeg{\"o}}, {\em Stability Theory of Dynamical
  Systems}, Springer-Verlag, 1970.

\bibitem{BicharaCCC2015}
{\sc D.~Bichara and C.~Castillo-Chavez}, {\em Vector-borne diseases models with
  residence times -- a lagrangian perspective}, Mathematical Biosciences, 281
  (2016), pp.~128--138.

\bibitem{bichara2016dynamics}
{\sc D.~Bichara, S.~A. Holechek, J.~Vel{\'a}zquez-Castro, A.~L. Murillo, and
  C.~Castillo-Chavez}, {\em On the dynamics of dengue virus type 2 with
  residence times and vertical transmission}, Letters in Biomathematics, 3
  (2016), pp.~140--160.

\bibitem{bichara2015sis}
{\sc D.~Bichara, Y.~Kang, C.~Castillo-Chavez, R.~Horan, and C.~Perrings}, {\em
  Sis and sir epidemic models under virtual dispersal}, Bulletin of
  mathematical biology, 77 (2015), pp.~2004--2034.

\bibitem{jmb}
B.~Bonzi, A.~Fall, A.~Iggidr, and G.~Sallet.
\newblock Stability of differential susceptibility and infectivity epidemic
  models.
\newblock {\em Journal of Mathematical Biology}, 62(1):39--64, 2011.


\bibitem{BlytheCCC89}
{\sc S.~P. Blythe and C.~Castillo-Chavez}, {\em Like-with-like preference and
  sexual mixing models}, Math. Biosci., 96 (1989), pp.~221--238.

\bibitem{castillo2016perspectives}
{\sc C.~Castillo-Chavez, D.~Bichara, and B.~R. Morin}, {\em Perspectives on the
  role of mobility, behavior, and time scales in the spread of diseases},
  Proceedings of the National Academy of Sciences, 113 (2016),
  pp.~14582--14588.

\bibitem{CCCBusen91}
{\sc C.~Castillo-Chavez and S.~Busenberg}, {\em A general solution of the
  problem of mixing of subpopulations and its application to risk-and
  age-structured epidemic models for the spread of aids}, Mathematical Medecine
  and Biology, 8 (1991), pp.~1--29.

\bibitem{CasThi95}
{\sc C.~Castillo-Chavez and H.~R. Thieme}, {\em Asymptotically autonomous
  epidemic models}, in Mathematical Population Dynamics: Analysis of
  Heterogeneity, Volume One: Theory of Epidemics,, O.~Arino, A.~D.E., and
  M.~Kimmel, eds., Wuerz, 1995.

\bibitem{cosner2009effects}
{\sc C.~Cosner, J.~Beier, R.~Cantrell, D.~Impoinvil, L.~Kapitanski, M.~Potts,
  A.~Troyo, and S.~Ruan}, {\em The effects of human movement on the persistence
  of vector-borne diseases}, Journal of theoretical biology, 258 (2009),
  pp.~550--560.

\bibitem{MR1057044}
{\sc O.~Diekmann, J.~A.~P. Heesterbeek, and J.~A.~J. Metz}, {\em On the
  definition and the computation of the basic reproduction ratio {$R\sb 0$} in
  models for infectious diseases in heterogeneous populations}, J. Math. Biol.,
  28 (1990), pp.~365--382.

\bibitem{dushoff1995effects}
{\sc J.~Dushoff and S.~Levin}, {\em The effects of population heterogeneity on
  disease invasion}, Mathematical biosciences, 128 (1995), pp.~25--40.

\bibitem{eckenrode2014association}
{\sc S.~Eckenrode, A.~Bakullari, M.~L. Metersky, Y.~Wang, M.~M. Pandolfi,
  D.~Galusha, L.~Jaser, and N.~Eldridge}, {\em The association between age,
  sex, and hospital-acquired infection rates: results from the 2009-2011
  national medicare patient safety monitoring system}, Infection Control \&
  Hospital Epidemiology, 35 (2014), pp.~S3--S9.

\bibitem{falcon2016day}
{\sc J.~A. Falc{\'o}n-Lezama, R.~A. Mart{\'\i}nez-Vega, P.~A. Kuri-Morales,
  J.~Ramos-Casta{\~n}eda, and B.~Adams}, {\em Day-to-day population movement
  and the management of dengue epidemics}, Bulletin of Mathematical Biology, 78
  (2016), pp.~2011--2033.

\bibitem{fall2007epidemiological}
{\sc A.~Fall, A.~Iggidr, G.~Sallet, and J.-J. Tewa}, {\em Epidemiological
  models and lyapunov functions}, Math. Model. Nat. Phenom, 2 (2007),
  pp.~62--68.

\bibitem{EliCCC2011}
{\sc E.~Fenichel, C.~Castillo-Chavez, M.~G. Ceddia, G.~Chowell,
  P.~Gonzalez~Parra, G.~J. Hickling, G.~Holloway, R.~Horan, B.~Morin,
  C.~Perrings, M.~Springborn, L.~Valazquez, and C.~Villalobos}, {\em Adaptive
  human behavior in epidemiological models}, PNAS,  (2011).

\bibitem{MR87c:92046}
{\sc H.~W. Hethcote and H.~R. Thieme}, {\em Stability of the endemic
  equilibrium in epidemic models with subpopulations}, Math. Biosci., 75
  (1985), pp.~205--227.

\bibitem{Hirsch84}
{\sc M.~Hirsch}, {\em The dynamical system approach to differential equations},
  Bull. AMS, 11 (1984), pp.~1--64.

\bibitem{horn1985matrix}
{\sc R.~A. Horn, C.~R. Johnson}, {\em Matrix Analysis}, Cambridge University Press, 1985), New York.

\bibitem{HuangCookeCCC92}
{\sc W.~Huang, K.~Cooke, and C.~Castillo-Chavez}, {\em Stability and
  bifurcation for a multiple-group model for the dynamics of hiv/aids
  transmission}, SIAM J. Appl. Math., 52 (1992), pp.~835--854.

\bibitem{iggidr2014dynamics}
{\sc A.~Iggidr, G.~Sallet, and M.~O. Souza}, {\em On the dynamics of a class of
  multi-group models for vector-borne diseases}, Journal of Mathematical
  Analysis and Applications, 2 (2016), pp.~723--743.

\bibitem{IggidrSalletTsanou}
{\sc A.~Iggidr, G.~Sallet, and B.~Tsanou}, {\em Global stability analysis of a
  metapopulation sis epidemic model}, Math. Pop. Stud., 19 (2012),
  pp.~115--129.

\bibitem{JacSimKoo95}
{\sc J.~A. Jacquez, C.~P. Simon, and J.~Koopman}, {\em Core groups and the r0s
  for subgroups in heterogeneous sis and si models}, in Epidemics models :
  their structure and relation to data, D.~ed., ed., Cambridge University
  Press, 1996, pp.~279--301.

\bibitem{Jacq88}
{\sc J.~A. Jacquez, C.~P. Simon, J.~Koopman, L.~Sattenspiel, and T.~Perry},
  {\em modeling and analyzing {HIV} transmission : the effect of contact
  patterns}, Math. Biosci., 92 (1988).

\bibitem{kaplan2002hospitalized}
{\sc V.~Kaplan, D.~C. Angus, M.~F. Griffin, G.~Clermont, R.~Scott~Watson, and
  W.~T. Linde-zwirble}, {\em Hospitalized community-acquired pneumonia in the
  elderly: age-and sex-related patterns of care and outcome in the united
  states}, American journal of respiratory and critical care medicine, 165
  (2002), pp.~766--772.

\bibitem{KmK1927}
{\sc W.~Kermack and A.~McKendrick}, {\em A contribution to the mathematical
  theory of epidemics}, Proc. R. Soc., A115 (1927), pp.~700--721.

\bibitem{LajYo76}
{\sc A.~Lajmanovich and J.~Yorke}, {\em A deterministic model for gonorrhea in
  a nonhomogeneous population.}, Math. Biosci., 28 (1976), pp.~221--236.

\bibitem{LaSLef61}
{\sc J.~P. LaSalle and S.~Lefschetz}, {\em Stability by Liapunov's direct
  method}, Academic Press, 1961.

\bibitem{metz2014dynamics}
{\sc J.~A. Metz and O.~Diekmann}, {\em The dynamics of physiologically
  structured populations}, vol.~68, Springer, 2014.

\bibitem{Nold80}
{\sc A.~Nold}, {\em Heterogeneity in disease-transmission modeling.}, Math.
  Biosci., 52 (1980), p.~227.

\bibitem{Perrings:2014yq}
{\sc C.~Perrings, C.~Castillo-Chavez, G.~Chowell, P.~Daszak, E.~P. Fenichel,
  D.~Finnoff, R.~D. Horan, A.~M. Kilpatrick, A.~P. Kinzig, N.~V. Kuminoff,
  S.~Levin, B.~Morin, K.~F. Smith, and M.~Springborn}, {\em Merging economics
  and epidemiology to improve the prediction and management of infectious
  disease}, Ecohealth,  (2014).

\bibitem{prothero1977disease}
{\sc R.~M. Prothero}, {\em Disease and mobility: a neglected factor in
  epidemiology}, International journal of epidemiology, 6 (1977), pp.~259--267.

\bibitem{rodriguez2001models}
{\sc D.~J. Rodr{\'\i}guez and L.~Torres-Sorando}, {\em Models of infectious
  diseases in spatially heterogeneous environments}, Bulletin of Mathematical
  Biology, 63 (2001), pp.~547--571.

\bibitem{Ruktanonchai2016}
{\sc N.~W. Ruktanonchai, D.~L. Smith, and P.~De~Leenheer}, {\em Parasite
  sources and sinks in a patched ross-macdonald malaria model with human and
  mosquito movement: implications for control}, Mathematical Biosciences, 279
  (2016), pp.~90--101.

\bibitem{rushton1955deterministic}
{\sc S.~Rushton and A.~Mautner}, {\em The deterministic model of a simple
  epidemic for more than one community}, Biometrika,  (1955), pp.~126--132.

\bibitem{SalVdd06}
{\sc M.~Salmani and P.~van~den Driessche}, {\em A model for disease
  transmission in a patchy environment}, DCDS series B, 6 (2006), pp.~185--202.

\bibitem{7606146}
{\sc L.~Sattenspiel and K.~Dietz}, {\em A structured epidemic model
  incorporating geographic mobility among regions.}, Math Biosci, 128 (1995),
  pp.~71--91.

\bibitem{SattenspielSimon88}
{\sc L.~Sattenspiel and C.~P. Simon}, {\em The spread and persistence of
  infectious diseases in structured populations}, Math. Biosci., 90 (1988),
  pp.~341--366.
\newblock Nonlinearity in biology and medicine (Los Alamos, NM, 1987).

\bibitem{VddWat02}
{\sc P.~van~den Driessche and J.~Watmough}, {\em reproduction numbers and
  sub-threshold endemic equilibria for compartmental models of disease
  transmission}, Math. Biosci.,  (2002), pp.~29--48.

\bibitem{0478.93044}
{\sc M.~Vidyasagar}, {\em {Decomposition techniques for large-scale systems
  with nonadditive interactions: Stability and stabilizability.}}, IEEE Trans.
  Autom. Control, 25 (1980), pp.~773--779.

\bibitem{xiao2014transmission}
{\sc Y.~Xiao and X.~Zou}, {\em Transmission dynamics for vector-borne diseases
  in a patchy environment}, Journal of mathematical biology, 69 (2014),
  pp.~113--146.

\bibitem{Yorke:1978wu}
{\sc J.~A. Yorke, H.~W. Hethcote, and A.~Nold}, {\em Dynamics and control of
  the transmission of gonorrhea.}, Sex Transm Dis, 5 (1978), pp.~51--56.


\end{thebibliography}

\end{document}